\documentclass[journal]{IEEEtran}
\usepackage{multirow}%added by yxw 20161022
\usepackage{diagbox}%added by yxw 20161022
\usepackage{amssymb}
\usepackage{amsmath}
\usepackage{graphicx}
\usepackage{cite}
\usepackage{citesort}
\usepackage{subfigure}
\usepackage{graphicx,epstopdf}
\usepackage{epsfig}	
\usepackage{cite,graphicx,amsmath,amssymb}
\usepackage{comment}

\usepackage{amssymb}
\usepackage{amsmath}
\usepackage{cite}
\usepackage{url}
\usepackage{xcolor}
\usepackage{cite,graphicx,amsmath,amssymb}
\usepackage{subfigure}
\usepackage{citesort}
\usepackage{fancyhdr}
\usepackage{mdwmath}
\usepackage{mdwtab}
\usepackage{caption}
\usepackage{amsthm}
\usepackage{lipsum}

\newtheorem{theorem}{Theorem}

\newtheorem{lemma}{Lemma}

\newtheorem{corollary}{Corollary}

\newtheorem{remark}{Remark}  % added by yxw 20160929
\makeatletter
\def\ScaleIfNeeded{%
\ifdim\Gin@nat@width>\linewidth \linewidth \else \Gin@nat@width
\fi } \makeatother

\begin{document}

\title{\Huge{Exploiting Full/Half-Duplex User Relaying in NOMA Systems}}

\author{ Xinwei~Yue,~\IEEEmembership{Student Member,~IEEE,} Yuanwei\ Liu,~\IEEEmembership{Member,~IEEE,}
 Shaoli~Kang, Arumugam~Nallanathan,~\IEEEmembership{Fellow,~IEEE},
 and Zhiguo~Ding,~\IEEEmembership{Senior Member,~IEEE}

\thanks{This work was supported by National High Technology Research and Development Program of China (863 Program, 2015AA01A709). The work of Z. Ding was supported by the UK EPSRC under grant number EP/L025272/1 and by H2020-MSCA-RISE-2015 under grant number 690750.
}
\thanks{X. Yue and S. Kang are with School of Electronic and Information Engineering, Beihang university, Beijing 100191,
China. S. Kang is also with State Key Laboratory of Wireless Mobile Communications, China Academy of Telecommunications Technology(CATT), Beijing 100094, China (email: xinwei$\_$yue@buaa.edu.cn, kangshaoli@catt.cn).}
\thanks{Y. Liu and A. Nallanathan are with School of Electronic Engineering and Computer Science, Queen Mary University of London, London E1 4NS, U.K. (email: \{yuanwei.liu, a.nallanathan\}@qmul.ac.uk).}
\thanks{Z. Ding is with the Department of Electrical Engineering, Princeton University, Princeton, USA and also with the School of Computing and Communications, Lancaster University, Lancaster LA1 4YW, U.K (e-mail: z.ding@lancaster.ac.uk). Part of this work has been submitted to IEEE ICC 2017 \cite{Yue2016Non}.}
 }

%Z. Ding is with School of Computing and Communications, Lancaster University,
 %UK (e-mail: z.ding@lancaster.ac.uk).
%\author{
%\IEEEauthorblockN{  Xinwei~Yue\IEEEauthorrefmark{1}, Yuanwei Liu\IEEEauthorrefmark{3}, Shaoli Kang\IEEEauthorrefmark{2},Arumugam Nallanathan\IEEEauthorrefmark{3}   }
% Yuanwei Liu
% Arumugam Nallanathan
%\IEEEauthorblockA{\\
%\IEEEauthorrefmark{1}School of Electronic and Information Engineering
%, Beihang University, Beijing, China\\
%\IEEEauthorrefmark{2} China Academy of Telecommunication Technology, Beijing, China\\
%\IEEEauthorrefmark{3} School of Information and Communication Engineering, Beijing University of Posts and Telecommunications, Beijing, %China\\
%\IEEEauthorrefmark{3} King's College London, London, UK
% } }

\maketitle

\begin{abstract}
%This paper focuses on the performance analysis for the full-duplex cooperative nor-orthogonal multiple access (FD NOMA) in which the direct link is applied to improve the reliability of far user.
In this paper, a novel cooperative non-orthogonal multiple access (NOMA) system is proposed, where one near user is employed as decode-and-forward (DF) relaying switching between full-duplex (FD) and half-duplex (HD)  mode to help a far user. Two representative cooperative relaying scenarios are investigated insightfully. The \emph{first scenario} is that no direct link exists between the base station (BS) and far user. The \emph{second scenario} is that the direct link exists between the BS and far user. To characterize the performance of potential gains brought by FD NOMA in two considered scenarios, three performance metrics outage probability, ergodic rate and energy efficiency are discussed. More particularly, we derive new closed-form expressions for both exact and asymptotic outage probabilities as well as delay-limited throughput for two NOMA users. Based on the derived results, the diversity orders achieved by users are obtained. We confirm that the use of direct link overcomes zero diversity order of far NOMA user inherent to FD relaying. Additionally, we derive new closed-form expressions for asymptotic ergodic rates. Based on these, the high signal-to-noise radio (SNR) slopes of two users for FD NOMA are obtained. Simulation results demonstrate that: 1) FD NOMA is superior to HD NOMA in terms of outage probability and ergodic sum rate in the low SNR region; and 2) In delay-limited transmission mode, FD NOMA has higher energy efficiency than HD NOMA in the low SNR region; However, in delay-tolerant transmission mode, the system energy efficiency of HD NOMA exceeds FD NOMA in the high SNR region.
\end{abstract}
\section{Introduction}
With the rapid increasing demand of wireless networks, the requirements for efficiently exploiting  the spectrum is of great significance in new radio (NR) usage scenarios \cite{Cai2017Modulation}. To achieve higher spectral efficiency of the fifth generation (5G) mobile communication network, non-orthogonal multiple access (NOMA) has received a great deal of attention~ \cite{METIS}. Recently, several NOMA schemes have been researched in detail, such as power domain NOMA (PD-NOMA) \cite{Saito2013Non}, sparse code multiple access (SCMA) \cite{Nikopour2013Sparse}, pattern division multiple access (PDMA) \cite{ChenPattern7526461}, multiuser sharing access (MUSA) \cite{Yuan2016Multi}, etc. Generally speaking, NOMA schemes can be classified into two categories, namely power-domain NOMA\footnote{In this paper, we focus on power-domain NOMA and use NOMA to represent PD-NOMA.} and code-domain NOMA. Downlink multiuser superposition transmission (DL MUST), the special case of NOMA, has been studied for 3rd generation partnership project (3GPP) in \cite{MUST1}. The pivotal characteristic of NOMA is to allow multiple users to share the same physical resource (i.e., time/fequency/code) via different power levels. At the receiver side, the successive interference cancellation (SIC) is carried out \cite{Ding2015Application}.

So far, point-to-point NOMA has been studied extensively in \cite{Al6933459,Zhang7390209,Ding6868214,Pairing7273963,Liu2016TVT}. To evaluate the performance of uplink NOMA systems, the authors in \cite{Al6933459} proposed the uplink NOMA transmission scheme to achieve higher system rate. The expressions of outage probability and achievable sum rates for uplink NOMA were derived with a novel uplink control protocol in \cite{Zhang7390209}. Regarding downlink NOMA scenarios, authors in \cite{Ding6868214} analyzed the outage behavior and ergodic rates of NOMA networks, where multiple NOMA users are spatial randomly deployed in a disc. In \cite{Pairing7273963}, the cognitive radio inspired NOMA concept was proposed, in which the influence of user pairing with the fixed power allocation in NOMA systems was discussed. As the interplay between NOMA and cognitive radio is bidirectional, NOMA was also applied to cognitive radio networks in \cite{Liu2016TVT}. More particularly, the analytical expressions of outage probability was derived and diversity orders were characterized. Apart from the above works, a new opportunistic NOMA scheme was proposed in \cite{Tian7390804} to improve the efficiency of SIC. In \cite{Yang2016Outage}, the flexible power allocation mode was researched in terms of outage probability for hybrid NOMA systems. The quantum-assisted multiple users transmission mode for NOMA was proposed in \cite{Botsinis2016Quantum}, which utilizes the minimum bit error ratio criterion to optimize the predefined transmitted information.
The author of \cite{Choi7383267} has studied the linear MUST scheme for NOMA to maximize sum rate of the entire network.
With the emphasis on physical layer security, in \cite{Physical7812773}, authors have adopted two effective approaches, namely protected zone and artificial noise for enhancing the secrecy performance of NOMA networks with the aid of stochastic geometry.

Cooperative communication is a particularly effective approach by providing the higher diversity as well as extending the coverage of networks \cite{laneman2004cooperative}. Current NOMA research contributions in terms of cooperative communication mainly include two aspects. The \emph{first aspect} is the application of NOMA into cooperative networks \cite{Choi6708131,Kim7230246,Kim2015Capacity,Men7454773}. The coordinated two-point system with superposition coding (SC) was researched in the downlink communication in \cite{Choi6708131}. The authors in \cite{Kim7230246,Kim2015Capacity} investigated outage probability and system capacity of decode-and-forward (DF) relaying for NOMA. In \cite{Men7454773}, the outage behavior of amplify-and-forward (AF) relaying with NOMA has been discussed over Nakagami-$m$ fading channels. The \emph{second aspect} is cooperative NOMA, which was first proposed in \cite{Ding2014Cooperative}. The key idea of cooperative NOMA is to regard the near NOMA user as a DF user relaying to help far NOMA user. On the standpoint of considering energy efficiency issues, simultaneous wireless information and power transfer (SWIPT) was employed at the near NOMA user, which was regarded as DF relaying in \cite{Liu7445146SWIPT}.

Although cooperative NOMA is capable of enhancing the performance gains for far user, it results in additional bandwidth costs for the system. To avoid this issue, one promising solution is to adopt the full-duplex (FD) relay technology. FD relay receives and transmits simultaneously in the same frequency band, which is the reason why it has attracted significant interest to realize more spectrally efficient systems \cite{Ju2009Improving}. In a general case, due to the imperfect isolation or cancellation process, FD operation may suffer from residual loop self-interference (LI) which is modeled as a fading channel. With the development of signal processing and antenna technologies, relaying with FD operation is feasible \cite{Riihonen2011Mitigation}. Recently, FD relay technologies have been proposed as a promising technique for 5G networks in \cite{Zhang2015Full}. Two main types of FD relay techniques, namely FD AF relaying and FD DF relaying, have been discussed in \cite{Wang2015Outage,Exploiting7105858,Kwon2010Optimal}. The expressions for outage probability of FD AF relaying were provided in \cite{Wang2015Outage}, which considers the processing delay of relaying in practical scenarios. In \cite{Exploiting7105858}, the performance of FD AF relaying in terms of outage probability was investigated considering the direct link. The authors in \cite{Kwon2010Optimal} characterized the outage performance of FD DF relaying. It is demonstrated that the optimal duplex mode can be selected according to the outage probability. Furthermore, the operations of randomly switching between FD and HD mode were considered for enhancing spectral efficiency in \cite{Riihonen5961159Hybrid}.
\subsection{Motivations and Related Works}
While the aforementioned research contributions have laid a solid foundation with providing a good understanding of cooperative NOMA and FD relay technology, the treatises for investigating the potential benefits by integrating these two promising technologies are still in their infancy. Some related cooperative NOMA studies have been investigated in \cite{Ding2014Cooperative,Ding2016FD}. In \cite{Ding2014Cooperative}, it is demonstrated that
the maximum diversity order can be obtained for all users, but cooperative NOMA with a direct link was only considered with HD operation mode. In  \cite{Ding2016FD}, the authors investigated the performance of FD device-to-device based cooperative NOMA. However, only the outage performance of far user was analyzed. To the best of our knowledge, there is no existing work investigating the impact of the direct link for FD user relaying on the network performance, which motivates us to develop this treatise. Also, there is lack of systematic performance evaluation metrics i.e., considering ergodic rate and energy efficiency in terms of FD/HD NOMA systems. Different from \cite{Ding2014Cooperative,Ding2016FD}, we present a comprehensive investigation on adopting near user as a FD/HD relaying to improve the reliability of far user. More specifically, we attempt to explore the potential ability of user relaying in NOMA networks with identifying the following key impact factors.
\begin{itemize}
  \item Will FD NOMA relaying bring performance gains compared to HD NOMA relaying? If yes, what is the condition?
  \item What is the impact of direct link on the considered system? Will it significantly improve the network performance in terms of outage probability and throughput?
  \item Will NOMA relaying bring performance gains compared to conventional orthogonal multiple access (OMA) relaying?
  \item In delay-limited/tolerant transmission modes, what are the relationships between energy efficiency (EE) and HD/FD NOMA systems?
\end{itemize}

\subsection{Contributions}
%\textcolor[rgb]{0.00,0.00,1.00}{}
%In this paper, we consider a scenario in which the direct link between the BS and the weak user exists. In addition, no direct link scenario is also considered for FD/HD NOMA systems.
%The primary contributions of this paper are summarised as follows:
%Specifically, we strive to answer the following fundamental questions. %with the help of
In this paper, we propose a comprehensive NOMA user relaying system, where near user can switch between FD and HD mode according to the channel conditions. We also consider the setting of two scenarios in which the direct link exists or not between the BS and far user. Based on our proposed NOMA user relaying systems, the primary contributions of this paper are summarised as follows:

\begin{enumerate}
  %Will NOMA relaying bring performance gains compared to the conventional OMA relaying.
  \item \emph{Without direct link:} We derive the closed-form expressions of outage probability for the near user and far user, respectively. For obtaining more insights, we further derive the asymptotic outage probability of two users and obtain diversity orders at high SNR. We demonstrate that FD NOMA converges to an error floor and results in a zero diversity order. We show that FD NOMA is superior to HD NOMA in terms of outage probability in the low SNR region rather than in the high SNR region. In addition, we also obtain the diversity orders of two users for HD NOMA. Furthermore, we analyze the system throughput in delay-limited transmission according to the derived outage probability.
%Next, the asymptotic outage probability of two users are easily obtained at high SNR and the diversity orders achieved by the near user and far user can be attained. It is shown that FD NOMA converge to an error floor and obtain a zero diversity order. We conclude that FD NOMA is superior to HD NOMA in terms of outage probability at low SNR region rather than at high SNR region. The diversity order of one is achieved for two users in HD NOMA. In addition, we analyze the delay-limited transmission throughput according to the derived outage probability.
%We observe that FD NOMA relaying bring performance gains compared to HD NOMA and conventional OMA relaying at low SNR region.
%Will FD NOMA relaying bring performance gains compared to the HD NOMA relaying. If yes, what is the condition
  \item \emph{Without direct link:} We study the ergodic rate of two users for FD/HD NOMA. To gain better insights, we derive the asymptotic ergodic rates of two users and obtain the high SNR slopes. We demonstrate that the ergodic rate of far user converges to a throughput ceiling for FD/HD NOMA in the high SNR region. Moreover, we also demonstrate that FD NOMA outperforms HD NOMA in terms of ergodic sum rate in the low SNR region.
%We study the ergodic sum rate of FD/HD NOMA without direct link. The exact expression of ergodic rate for the near user is derived and approximate expression for far user at high SNR is given. Next, the high SNR slops of two users for FD/HD NOMA are obtained. It is seen that far user converges to throughput ceiling and obtains a zero high SNR slop. Observation is that FD NOMA outperforms the HD NOMA system in terms of ergodic sum rate.
%What is the impact of direct link on the considered system? Will it significantly improve the network performance in terms of outage probability and throughput
  \item \emph{With direct link:} We first derive the closed-form expression in terms of outage probability for far user. In order to get the corresponding diversity order, we also derive the approximated outage probability of far user.  We find that the reliability of far user is improved with the help of direct link. We confirm that the use of direct link overcomes the zero diversity order of far user inherent to conventional FD relaying.
      For the near user, the diversity order is the same as that of FD relaying. Additionally, we conclude that the superiority of FD NOMA is no longer apparent with the values of LI increasing.

%  We derive the investigate the outage performance of FD NOMA with direct link. The reliability of the far user is improved with the help of direct link and lower outage probability is achieved. The closed-form expression of outage probability for far user is derived and diversity order is obtained which is equal to one. Therefore, the use of the direct link overcomes the zero diversity order of far NOMA user inherent to conventional FD relaying. For the near user, zero diversity order is the same as FD relaying. It is worth noting that the superiority of FD NOMA is no longer apparent with the values of LI increasing.
 \item \emph{With direct link:} We analyze the ergodic rate of far user for FD/HD NOMA. For this scenario, it is the fact that the ergodic rate of near user is invariant which is not affected by the direct link. Similarly, we also derive the approximated expressions for ergodic rate and obtain the high SNR slopes. We demonstrate that the use of direct link is incapable of assisting far user to obtain additional high SNR slope. %We observe that the FD/HD NOMA system can work with a small rate in the low SNR region.
 \item \emph{Energy efficiency:} We derive expressions in terms of energy efficiency for FD/HD NOMA. We conclude that FD NOMA without/with direct link have a higher energy efficiency corresponding to HD NOMA in the low SNR region for delay-limited transmission mode. However, in delay-tolerant transmission mode, the system energy efficiency of HD NOMA exceeds FD NOMA without/with direct link.
\end{enumerate}

%In this work, the full-duplex cooperative NOMA (FD NOMA) for downlink communications scenario is investigated, which includes a near user with better channel conditions and far user with worse channel conditions. Naturally, the near user is considered as DF relaying to transmit the information for far user.To characterize the performance of FD NOMA, closed-form expressions for outage probability of two users are derived and the delay-limited transmission throughput is analyzed. Based on analytical results, the approximation expressions for outage probability at high SNR regime is obtained. On the condition of that there is direct link between the BS and the far user, the diversity order achieved by two users is zero and one, respectively. FD mode results in the zero diversity order at near user just as in \cite{Krikidis2012Full}. It is observe that FD NOMA is superior to half-duplex NOMA (HD NOMA) at the low SNR region instead of at high SNR region. With the LI increasing, the superiority of FD NOMA is no longer apparent. In addition, the ergodic rates for two users are discussed insightfully and the corresponding ergodic sum rate is obtained. The results show that FD NOMA outperforms the HD NOMA system in terms of ergodic sum rate.

\subsection{Organization and Notation}
The rest of the paper is organized as follows. In Section \ref{System Model}, the  system model of user relaying for FD NOMA is set up. In Section \ref{Section_III}, the analytical expressions for outage probability, diversity order and throughput of FD/HD user relaying are derived and analyzed. In Section \ref{ergodic rate}, the performance of user relaying for FD/HD NOMA are evaluated in terms of ergodic rate. Section \ref{Energy Efficiency} considers the system energy efficiency for FD/HD NOMA systems.
Analytical results and simulations are presented in Section \ref{Numerical Results}. Section \ref{Conclusion} concludes the paper.

The main notations of this paper is shown as follows:
$\mathbb{E}\{\cdot\}$ denotes expectation operation; ${f_X}\left(  \cdot  \right)$ and ${F_X}\left(  \cdot  \right)$ denote the probability density function (PDF) and the cumulative distribution function (CDF) of a random variable $X$; $ \propto $ represents ``be proportional to".

\section{System Model}\label{System Model}
\begin{figure}[t!]
    \begin{center}
        \includegraphics[width=3.4in,  height=1.6in]{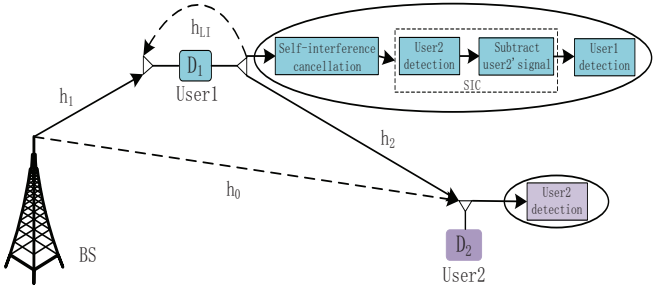}
        \caption{A downlink FD/HD cooperative NOMA system model.}
        \label{Fig. 1}
    \end{center}
\end{figure}

We consider a FD cooperative NOMA system consisting of one source, i.e, the BS, that intends to communicate with far user $D_{2}$ via the assistance of near user $D_{1}$  illustrated in Fig. 1.
$D_{1}$ is regarded as user relaying and DF protocol is employed to decode and forward the information to $D_{2}$.
To enable FD communication, $D_{1}$ is equipped with one transmit antenna and one receive antenna, while the BS and $D_{2}$ are single-antenna nodes. Note that $D_1$ can switch operation between FD and HD mode. All wireless links in network are assumed to be independent non-selective block Rayleigh fading and are disturbed by additive white Gaussian noise with mean power $N_{0}$. ${h_{1}}$, ${h_{2}}$, and ${h_{0}}$ are denoted as the complex channel coefficient of $BS \rightarrow D_{1}$, $D_{1} \rightarrow D_{2}$, and $BS \rightarrow D_{2}$ links, respectively. The channel power gains ${|h_{1}}|^2$, ${|h_{2}}|^2$ and ${|h_{0}}|^2$ are assumed to be exponentially distributed random variables (RVs) with the parameters $ \Omega_{i} $, $\emph{i}\in\{1,2,0\}$, respectively. When $D_{1}$ operates in FD mode, we assume that an imperfect self-interference cancellation scheme\footnote{LI refers to the signals that are transmitted by a FD
relaying and looped back to the receiver simultaneously. Through radio frequency (RF) cancellation, antenna cancellation and signal process technologies, etc, those LI can be suppressed to a lower level. However, LI still remains in the receiver due to imperfect self-interference cancellation, when decoding the desired signal.} is executed at $D_{1}$ such as in\cite{Exploiting7105858,Krikidis2012Full}. The LI is modeled as a Rayleigh fading channel with coefficient $h_{LI}$, and $ \Omega_{LI} $ is the corresponding average power.
To analyze HD NOMA, we introduce the switching operation factor detailed in the following.

During the $k$-th time slot, according to \cite{Ding6868214}, $D_{1}$ receives the superposed signal and loop interference signal simultaneously. The observation at $D_{1}$ is given by
\begin{align}\label{D1}
 {y_{{D_1}}}\left[ k \right] =& {h_1}(\sqrt {{a _1}{P_s}} {x_1}\left[ k \right] + \sqrt {{a _2}{P_s}} {x_2}\left[ k \right]) \nonumber \\
  &+ {h_{LI}}\sqrt {\varpi {P_r}} {x_{LI}}\left[ {k - \tau } \right] + {n_{{D_1}}}\left[ k \right],
\end{align}
where $\varpi$ is the switching operation factor between FD and HD mode. $\varpi=1$ and $\varpi=0$ denote $D_{1}$ working in FD and HD mode, respectively. Based on the practical application scenarios, we can select the different operation mode. ${x_{LI}}\left[ {k - \tau } \right]$ denotes loop interference signal and $\tau $ denotes the processing delay at $D_{1}$ with an integer $\tau  \ge 1$. More particularly, we assume that the time $k$ satisfies the relationship $k \ge \tau$. $P_{s}$ and $P_{r}$ are the normalized transmission powers at the BS and $D_{1}$, respectively. $x_1$ and $x_2$ are the signals for $D_1$ and $D_2$, respectively. ${a_{1}}$ and ${a_{2}}$ are the corresponding power allocation coefficients. To stipulate better fairness between the users, we assume that $a_{2} > a_{1}$ with $a_{1}+a_{2}=1$. The SIC\footnote{It is assumed that perfect SIC is employed at $D_1$, our future work will relax this ideal assumption.} \cite{Cover1991Elements} can be invoked by $D_1$ for first detecting $D_2$ having a larger transmit power, which has less inference signal. Then the signal of $D_2$ can be detected from the superposed signal. Therefore, the received signal-to-interference-plus-noise ratio (SINR) at $D_{1}$ to detect $D_{2}$'s message $x_{2}$ is given by
\begin{align}\label{SINR for D2 to detect D1}
{\gamma _{{D_2} \to {D_1}}} = \frac{{{{\left| {{h_1}} \right|}^2}{a_2}\rho }}{{{{\left| {{h_1}} \right|}^2}{a_1}\rho  + \varpi {{\left| {{h_{LI}}} \right|}^2}\rho  + 1}},
\end{align}
where ${\rho}=\frac{P_{s}}{N_{0}}$ is the transmit signal-to-noise radio (SNR). Note that $x_{1}$ and $x_{2}$ are supposed to be normalized unity power signals, i.e, $\mathbb{E}\{x_{1}^2\}= \mathbb{E}\{x_{2}^2\}=1$.

After SIC, the received SINR at $D_{1}$ to detect its own message $x_{1}$ is given by
\begin{align}\label{SINR for D2 at D1}
{\gamma _{{D_1}}} = \frac{{{{\left| {{h_1}} \right|}^2}{a_1}\rho }}{{\varpi {{\left| {{h_{LI}}} \right|}^2}\rho  + 1}}.
\end{align}

In the FD mode, the received signal at $D_{2}$ is written as
${y_{{D_2}}}\left[ k \right] = {h_0}(\sqrt {{a _1}{P_s}} {x_1}\left[ k \right] + \sqrt {{a _2}{P_s}} {x_2}\left[ k \right]) + {h_2}\sqrt {{P_r}} {x_2}\left[ {k - \tau } \right] + {n_{{D_2}}}\left[ k \right]$.
However, the observation at $D_{2}$ for the direct link is written as
%$\emph{y}_{D_{2}} = h_{0}{\left( {\sqrt {a_{1}P_{s} }x_{1} }+{\sqrt {a_{2} P_{s} }x_{2} } \right)} + n_{D_{2}}.$
${y_{{1,D_2}}}\left[ k \right] = {h_0}(\sqrt {{a _1}{P_s}} {x_1}\left[ k \right] + \sqrt {{a _2}{P_s}} {x_2}\left[ k \right]) + {n_{{D_2}}}\left[ k \right]$. Due to the existence of residue interference (RI) from relaying link, the received SINR at $D_{2}$ to detect $x_{2}$ for direct link is given by
\begin{align}\label{SINR for 1D2 adding interference}
\gamma _{1,{D_2}}^{RI} = \frac{{{{\left| {{h_0}} \right|}^2}{a_2}\rho }}{{{{\left| {{h_0}} \right|}^2}{a_1}\rho  + \kappa {{\left| {{h_2}} \right|}^2}\rho  + 1}},
\end{align}
where $\kappa $ denotes the impact levels of RI. Since DF relaying protocol is invoked in $D_{1}$, we assume that $D_1$ can decode and forward the signal $x_{2}$ to $D_{2}$ successfully for relaying link from $D_{1}$ to $D_{2}$. As a consequence, the observation at $D_{2}$ for relaying link is written as ${y_{2,{D_2}}}\left[ k \right] = {h_2}\sqrt {{P_r}} {x_2}\left[ {k - \tau } \right] + {n_{{D_2}}}\left[ k \right]$. Similarly, considering the impact of RI from direct link, the received SINR at $D_{2}$ to $x_2$ for relaying link is given by
\begin{align}\label{SINR for 2D2 adding interference}
\gamma _{2,{D_2}}^{RI} = \frac{{{{\left| {{h_2}} \right|}^2}\rho }}{{\kappa {{\left| {{h_0}} \right|}^2}\rho  + 1}}.
\end{align}

As stated in \cite{Riihonen5961159Hybrid,Exploiting7105858}, the relaying link corresponding to direct link from BS to $D_{2}$ has small time delay for any transmitted signals. In other words, there is some temporal separation between the signal from $D_1$ and BS. To derive the theoretical results for practical NOMA systems, we assume that these signals from $D_1$ and BS are fully resolvable by $D_2$ \cite{Ding2016FD}. Hence, we provide the upper bounds of \eqref{SINR for 1D2 adding interference} and \eqref{SINR for 2D2 adding interference} in the following parts, which are the received SINRs at $D_{2}$ to detect $x_{2}$ for direct link and relaying link, i.e.
\begin{align}\label{SINR for 1D2}
{\gamma_{1,D_{2}}} = \frac{{{\left| {{h_{0}}} \right|}^2{a_{2}}{\rho}}}{{ {\left| {{h_{0}}} \right|}^2{a_{1}}{\rho} + 1}},
\end{align}
and
\begin{align}\label{SINR for 2D2}
{\gamma_{2,D_{2}}} = {{\left| {{h_{2}}} \right|}^2{\rho}},
\end{align}
respectively. At this moment, the signals from the relaying link and direct link are combined by maximal ratio combining~ (MRC) at $D_{2}$. So the received SINR after MRC at $D_{2}$ is given by
\begin{align}\label{MRC SINR for D2}
{\gamma^{MRC}_{D_{2}}} = {{\left| {{h_{2}}} \right|}^2{\rho}}+\frac{{{\left| {{h_{0}}} \right|}^2{a_{2}}{\rho}}}{{ {\left| {{h_{0}}} \right|}^2{a_{1}}{\rho} + 1}} .
\end{align}
%\section{Full-Duplex Cooperative NON-orthogonal multiple access}\label{Section_III}

\section{Outage Probability}\label{Section_III}
When the target rate of users is determined by its quality of service (QoS), the outage probability is an important metric for performance evaluation. We will evaluate the outage performance in two representative scenarios in the following.
\subsection{User Relaying without Direct Link}
In this subsection, the first scenario is investigated in terms of outage probability.
\subsubsection{Outage Probability of $D_{1}$}
According to NOMA protocol, the complementary events of outage at $D_{1}$ can be explained as: $D_{1}$ can detect $x_{2}$ as well as its own message $x_{1}$. From the above description, the outage probability of $D_{1}$ is expressed as
\begin{align}\label{OP express for D1 without directlink}
P_{{D_1}}^{FD} = 1 - {{\mathop{\rm P}\nolimits} _{\rm{r}}}\left( {{\gamma _{{D_2} \to {D_1}}} > \gamma _{t{h_2}}^{FD},{\gamma _{{D_1}}} > \gamma _{t{h_1}}^{FD}} \right),
\end{align}
%According to the NOMA principle, The outage events at $D_{1}$ can be explained as follows: One of which is that $D_{1}$ cannot detect $x_{2}$. Another is that $D_{1}$ cannot detect its own message $x_{1}$ on the conditions that $D_{1}$ can detect $x_{2}$, successfully. From the above description, the outage probability of $D_{1}$ can be given by
%\begin{align}\label{OP express for D1}
%&  {P_{{D_1}}} = {{\rm{P}}_{\rm{r}}}\left( {\frac{{{{\left| {{h_1}} \right|}^2}{\alpha _2}\rho }}{{{{\left| {{h_1}} \right|}^2}{\alpha _1}\rho  + \varpi {{\left| {{h_{LI}}} \right|}^2}\rho  + 1}} < {\gamma _{t{h_2}}}} \right) \nonumber\\
%&  + {{\rm{P}}_{\rm{r}}}\left( {{{\left| {{h_1}} \right|}^2}{\alpha _1}\rho  < {\gamma _{t{h_1}}},\frac{{{{\left| {{h_1}} \right|}^2}{\alpha _2}\rho }}{{{{\left| {{h_1}} \right|}^2}{\alpha _1}\rho  + \varpi {{\left| {{h_{LI}}} \right|}^2}\rho  + 1}} > {\gamma _{t{h_2}}}} \right),
%\end{align}
where $\varpi  = 1$. $\gamma _{t{h_1}}^{FD}=2^{R_{1}}-1$ with $R_{1}$ being the target rate at $D_{1}$ to detect $x_{1}$ and $\gamma _{t{h_2}}^{FD}=2^{R_{2}}-1$ with $R_{2}$ being the target rate at $D_{1}$ to detect $x_{2}$.

The following theorem provides the outage probability of $D_{1}$ for FD NOMA.
\begin{theorem} \label{theorem:1}
The closed-form expression for the outage probability of $D_{1}$ is given by
\begin{align}\label{OP derived for D1 without directlink}
P_{{D_1}}^{FD}{\rm{ = 1}} - \frac{{{\Omega _1}}}{{{\Omega _1} + \rho \varpi {\theta _1}{\Omega _{LI}}}}{e^{ - \frac{{{\theta _1}}}{{{\Omega _1}}}}},
\end{align}
where $\varpi=1$. ${\theta_{1}  = \max \left( {\tau_{1} ,\beta_{1} } \right)}$, ${\tau_{1}  = \frac{{{\gamma _{t{h_2}}^{FD}}}}{{\rho \left( {{a_2} - {a _1}{\gamma _{t{h_2}}^{FD}}} \right)}}}$ and ${\beta_{1}}  = \frac{{{\gamma _{t{h_1}}^{FD}}}}{{{a_1}\rho }}$.
Note \eqref{OP derived for D1 without directlink} is derived on the condition of ${a_2} > {a_1}{\gamma_{t{h_2}}^{FD}}$.
\begin{proof} By definition, ${J_1}$ denotes the complementary event at $D_{1}$ and is calculated as
\begin{align}\label{OP J1 derived for D1}
 {J_1} =& {{\rm{P}}_{\rm{r}}}\left( {{{\left| {{h_1}} \right|}^2} \ge \left( {\varpi {{\left| {{h_{LI}}} \right|}^2}\rho  + 1} \right){\theta _1}} \right) \nonumber\\
  = &\int_0^\infty  {\int_{\left( {x\varpi \rho  + 1} \right){\theta _1}}^\infty  {{f_{{{\left| {{h_{LI}}} \right|}^2}}}\left( x \right){f_{{{\left| {{h_1}} \right|}^2}}}\left( y \right)} } dxdy  \nonumber\\
  = &\frac{{{\Omega _1}}}{{{\Omega _1} + \rho \varpi {\theta _1}{\Omega _{LI}}}}{e^{ - \frac{{{\theta _1}}}{{{\Omega _1}}}}}.
\end{align}
Substituting  \eqref{OP J1 derived for D1} into \eqref{OP express for D1 without directlink}, \eqref{OP derived for D1 without directlink} can be obtained and the proof is completed.
\end{proof}
\end{theorem}
\begin{corollary} \label{corollary1 no directlink for HD NOMA }
Based on \eqref{OP derived for D1 without directlink}, the outage probability of $D_{1}$ for HD NOMA with $\varpi  = 0$ is given by
\begin{align}\label{corollary1 derived for D1 no directlink for HD NOMA }
P_{{D_1}}^{HD} = 1 - {{\rm{e}}^{ - \frac{{{\theta _2}}}{{{\Omega _{\rm{1}}}}}}},
\end{align}
where $\gamma _{t{h_1}}^{HD} = {2^{2{R_1} - 1}}$ and $\gamma _{t{h_2}}^{HD} = {2^{2{R_2} - 1}}$ denote the target SNRs at $D_{1}$ to detect $x_{1}$ and $x_{2}$ with HD mode, respectively. $\theta _2  = \max ({\tau _2},{\beta _2})$, ${\beta _2} = \frac{{\gamma _{t{h_1}}^{HD}}}{{{\alpha _1}\rho }}$ and ${\tau _2} = \frac{{\gamma _{t{h_2}}^{HD}}}{{\rho ({a_2} - {a _1}\gamma _{t{h_2}}^{HD})}}$ with ${a_2} > {a_1}\gamma _{t{h_2}}^{HD}$.
\end{corollary}
\subsubsection{Outage Probability of $D_{2}$}   %one of which
The outage events of $D_{2}$ can be explained as below. The first is that $D_{1}$ cannot detect $x_{2}$. The second is that $D_{2}$ cannot detect its own message $x_{2}$ on the conditions that $D_{1}$ can detect $x_{2}$ successfully. Based on these, the outage probability of $D_{2}$ is expressed as
\begin{align}\label{OP express for D2 without directlink}
P_{{D_2},nodir}^{FD} =& {{\rm{P}}_{\rm{r}}}\left( {{\gamma _{{D_2} \to {D_1}}} < {\gamma _{t{h_2}}^{FD}}} \right) \nonumber \\
  &+ {{\rm{P}}_{\rm{r}}}\left( {{\gamma _{2,{D_2}}} < {\gamma _{t{h_2}}^{FD}},{\gamma _{{D_2} \to {D_1}}} > {\gamma _{t{h_2}}^{FD}}} \right) ,
\end{align}
where $\varpi  = 1$.

The following theorem provides the outage probability of $D_{2}$ for FD NOMA.
\begin{theorem} \label{theorem:2}
The closed-form expression for the outage probability of $D_{2}$ without direct link is given by
\begin{align}\label{exect OP for D2 no directlink}
P_{{D_2},nodir}^{FD} = 1 - \frac{{{\Omega _1}}}
{{ {{\Omega _1} +  \rho \varpi \tau_{1} {\Omega _{{\text{LI}}}}} }} {e^{ - \left( \frac{\tau_{1} }
{{{\Omega _1}}}+ {\frac{{{\gamma _{t{h_2}}^{FD}}}}
{{\rho {\Omega _2} }} } \right)}},
\end{align}
where $\varpi=1$.
\begin{proof}
By definition, ${J_2}$ and $J_3$ denote the first and second outage events, respectively. The process calculated is given by
\begin{align}\label{derived process for D2 without directlink J2}
 {J_2} =& {{\rm{P}}_{\rm{r}}}\left( {{{\left| {{h_1}} \right|}^2} < \tau_{1} \left( {\varpi {{\left| {{h_{LI}}} \right|}^2}\rho  + 1} \right)} \right) \nonumber \\
  =& \int_0^\infty  {\int_0^{\tau_{1} \left( {\varpi y\rho  + 1} \right)} {{f_{{{\left| {{h_1}} \right|}^2}}}\left( x \right)} } {f_{{{\left| {{h_{LI}}} \right|}^2}}}\left( y \right)dxdy  \nonumber \\
  =& 1 - \frac{{{\Omega _1}}}{{{{\Omega _1} + \rho \varpi \tau_{1} {\Omega _{{\rm{LI}}}}} }}{e^{ - \frac{\tau_{1} }{{{\Omega _1}}}}}.
\end{align}
Applying some algebraic manipulations, $J_3$ is given by:
\begin{align}\label{derived process for D2 without directlink J3}
{J_3} = \frac{{{\Omega _1}}}{{ {{\Omega _1} + \rho \varpi \tau_{1} {\Omega _{{\rm{LI}}}}} }}{e^{ - \frac{\tau_{1} }{{{\Omega _1}}}}}\left( {1 - {e^{ - \frac{{{\gamma _{t{h_2}}^{FD}}}}{{\rho {\Omega _2} }}}}} \right).
\end{align}
Combining \eqref{derived process for D2 without directlink J2} and \eqref{derived process for D2 without directlink J3}, \eqref{exect OP for D2 no directlink} can be obtained and the proof is completed.
\end{proof}
\end{theorem}
\begin{corollary} \label{corollary2 no directlink for HD NOMA }
Based on \eqref{exect OP for D2 no directlink}, the outage probability of $D_{2}$ without direct link for HD NOMA with $\varpi  = 0$ is given by
\begin{align}\label{corollary1 derived for D2 no directlink for HD NOMA }
P_{{D_2},nodir}^{HD}  = 1 - {e^{ - \frac{{{\tau _2}}}{{{\Omega _1}}} - \frac{{\gamma _{t{h_2}}^{HD}}}{{\rho {\Omega _2}}}}}.
\end{align}
\end{corollary}
\subsubsection{Diversity Analysis}
To get more insights, the asymptotic diversity analysis is provided in terms of  outage probability investigated in high SNR region. The diversity order is defined as
\begin{align}\label{diversity order}
d =  - \mathop {\lim }\limits_{\rho  \to \infty } \frac{{\log \left( {P_D^\infty \left( \rho  \right)} \right)}}{{\log \rho }}.
\end{align}

\paragraph{$D_{1}$ for FD NOMA case}
Based on analytical result in \eqref{OP derived for D1 without directlink}, when $ \rho\ $$\rightarrow$ $\infty$, the asymptotic outage probability of $D_{1}$ for FD NOMA with ${e^{-x}}\approx 1-x$ is given by
\begin{align}\label{appro OP for user1 without directlink}
P_{{D_1}}^{FD,\infty }= {\rm{1}} - \frac{{{\Omega _1}}}{{{\Omega _1} + \rho {\theta _1}{\Omega _{LI}}}}.
\end{align}
Substituting \eqref{appro OP for user1 without directlink} into \eqref{diversity order}, we can obtain $d_{{D_1}}^{FD}= 0$.
\begin{remark}
The diversity order of $D_{1}$ is zero, which is the same as the conventional FD relaying.
\end{remark}
\paragraph{$D_{1}$ for HD NOMA case}
Based on analytical result in \eqref{corollary1 derived for D1 no directlink for HD NOMA }, the asymptotic outage probability of $D_{1}$ for HD NOMA is given by
\begin{align}\label{appro OP of user1 for HD without directlink}
P_{{D_1}}^{HD,\infty } = \frac{{{\theta _2}}}{{{\Omega _{\rm{1}}}}} \propto \frac{1}{\rho }.
\end{align}
Substituting \eqref{appro OP of user1 for HD without directlink} into \eqref{diversity order}, we can obtain $d_{{D_1}}^{HD} = 1$.
\paragraph{$D_{2}$ for FD NOMA case}
Based on \eqref{exect OP for D2 no directlink}, the asymptotic outage probability of $D_{2}$ for FD NOMA is given by
\begin{align}\label{appro OP for user2 without directlink}
P_{{D_2},nodir}^{FD,\infty }= 1 - \frac{{{\Omega _1}{\Omega _2}\rho  - {\Omega _1}{\gamma _{t{h_2}}^{FD}} - \tau_{1} \rho {\Omega _2}}}{{{\Omega _2}\rho \left( {{\Omega _1} + \tau_{1} \rho {\Omega _{{\rm{LI}}}}} \right)}}.
\end{align}
Substituting \eqref{appro OP for user2 without directlink} into \eqref{diversity order}, we can obtain $d_{{D_2},nodir}^{FD} = 0$.
\begin{remark}
The diversity order of $D_{2}$ is zero, which is the same as $D_{1}$ in FD NOMA.
\end{remark}
\paragraph{$D_{2}$ for HD NOMA case}
 Based on \eqref{corollary1 derived for D2 no directlink for HD NOMA }, the asymptotic outage probability of $D_{2}$ for HD NOMA is given by
\begin{align}\label{appro OP of user2 for HD without directlink }
P_{{D_2},nodir}^{HD,\infty } = \frac{{{\gamma _{t{h_2}}^{HD}}}}{{\rho {\Omega _2}}} + \frac{\tau_{2} }{{{\Omega _1}}} \propto \frac{1}{\rho }.
\end{align}

Substituting \eqref{appro OP of user2 for HD without directlink } into \eqref{diversity order}, we can obtain $d_{{D_2},nodir}^{HD} = 1$.

\begin{remark}\label{error floor existent for D1 and D2 without direct link}
As can be observed that $P_{{D_1}}^{FD,\infty }$ and $P_{{D_2},nodir}^{FD,\infty }$ are a constant independent of $\rho $, respectively. Substituting \eqref{appro OP for user1 without directlink} and \eqref{appro OP for user2 without directlink} into \eqref{diversity order}, we see that there are the error floors for outage probability of two users.
\end{remark}

\subsubsection{Throughput Analysis}\label{SecitonIII-A-4}
In this subsection, the delay-limited transmission mode \cite{Liu7445146SWIPT,Zhong6898012} is considered for FD/HD NOMA.
\paragraph{FD NOMA case}
In this mode, the BS transmits information at a constant rate $R$, which is subject to the effect of outage probability due to wireless fading channels. The system throughput of FD NOMA without direct link is given by
\begin{align}\label{Delay-limited Transmission nodirect for FD}
R_{l\_nodir}^{FD}= \left( {1 - {P_{{D_1}}^{FD}}} \right){R_{{1}}} + \left( {1 - P_{{D_2},nodir}^{FD}} \right){R_{{2}}},
\end{align}
where $P_{{D_1}}^{FD}$ and $P_{{D_2},nodir}^{FD}$ are given in \eqref{OP derived for D1 without directlink} and \eqref{exect OP for D2 no directlink}, respectively.
\paragraph{HD NOMA case}
Similar to \eqref{Delay-limited Transmission nodirect for FD}, the system throughput of HD NOMA without direct link is given by
\begin{align}\label{Delay-limited Transmission nodirect for HD}
R_{l\_nodir}^{HD} = \left( {1 - P_{{D_1}}^{HD}} \right){R_1} + \left( {1 - P_{{D_2},nodir}^{HD}} \right){R_2},
\end{align}
where ${P_{{D_1}}^{HD}}$ and $P_{{D_2},nodir}^{HD}$ are given in \eqref{corollary1 derived for D1 no directlink for HD NOMA } and \eqref{corollary1 derived for D2 no directlink for HD NOMA }, respectively.

\subsection{User Relaying with Direct Link}
In this subsection, we explore a more challenging scenario, where the direct link between the BS and $D_{2}$ is used to convey information and system reliability can be improved. However, the outage probability of $D_{1}$ will not be affected by the direct link. As such, we only show outage probability of $D_{2}$ in the following.
\subsubsection{Outage Probability of $D_{2}$}
For the second scenario, the outage events of $D_{2}$ for FD NOMA is described as below. One of the events is when $x_{2}$ can be detected at $D_{1}$, but the received SINR after MRC at $D_{2}$ in one slot is less than its target SNR. Another event is that neither $D_{1}$ nor $D_{2}$ can detect $x_{2}$. Therefore, the outage probability of $D_{2}$ is expressed as
\begin{align} \label{express OP for D2 adding RI}
P_{{D_2},dir}^{FD,RI} =& {{\rm{P}}_{\rm{r}}}\left( {\gamma _{1,{D_2}}^{RI} + \gamma _{2,{D_2}}^{RI} < \gamma _{t{h_2}}^{FD},{\gamma _{{D_2} \to {D_1}}} > \gamma _{t{h_2}}^{FD}} \right) \nonumber\\
  &+ {{\rm{P}}_{\rm{r}}}\left( {{\gamma _{{D_2} \to {D_1}}} < \gamma _{t{h_2}}^{FD},\gamma _{1,{D_2}}^{RI} < \gamma _{t{h_2}}^{FD}} \right).
\end{align}
Unfortunately, the closed-form expression of \eqref{express OP for D2 adding RI} for $D_2$ can not be derived successfully. However, it can be evaluated by using numerical simulations. To further obtain a theoretical result for $D_2$, exploiting the upper bounds of received SINRs derived in \eqref{SINR for 1D2} and \eqref{SINR for 2D2}, the outage probability of $D_{2}$ is expressed as
\begin{align} \label{express OP for D2}
P_{{D_2},dir}^{FD} = &{{\text{P}}_{\text{r}}}\left( {\gamma _{{D_2}}^{MRC} < {\gamma _{t{h_2}}^{FD}},{\gamma _{{D_2} \to {D_1}}} > {\gamma _{t{h_2}}^{FD}}} \right)\nonumber\\% 换行
&+ {{\text{P}}_{\text{r}}}\left( {{\gamma _{{D_2} \to {D_1}}} < {\gamma _{t{h_2}}^{FD}},{\gamma _{1,{D_2}}} < {\gamma _{t{h_2}}^{FD}}} \right),
\end{align}
where $\varpi  = 1$.

The following theorem provides the outage probability of $D_{2}$ for FD NOMA.
\begin{theorem} \label{theorem_outage:2}
The closed-form expression for the outage probability of $D_{2}$ with direct link is given by
\begin{align}\label{OP derived for D2}
 &P_{{D_2},dir}^{FD} = \left\{ {1 - {e^{ - \frac{{{\tau _1}}}{{{\Omega _0}}}}} - \sum\limits_{n = 0}^\infty  {\frac{{{{\left( { - 1} \right)}^n}{e^\varphi }}}{{n!{\phi _2}^{n + 1}}}} } \right.\left[ {\frac{{{{\left( { - 1} \right)}^{2n + 1}}{\phi _1}^{n + 1}}}{{(n + 1)!}}} \right. \nonumber\\
 &\left. { \times \left( {{\mathop{\rm Ei}\nolimits} \left( \psi  \right) - {\mathop{\rm Ei}\nolimits} \left( {{\phi _1}} \right)} \right)\left. { + \sum\limits_{k = 0}^n {\frac{{{{\left( {1 + {a _1}\rho {\tau _1}} \right)}^{n + 1}}{e^\psi }{\psi ^k} - {e^{{\phi _1}}}{\phi _1}^k}}{{\left( {n + 1} \right)n \cdots \left( {n + 1 - k} \right)}}} } \right]} \right\} \nonumber\\
  &\times \chi {e^{ - \frac{{{\tau _1}}}{{{\Omega _1}}}}} + \left[ {\left( {1 - \chi {e^{ - \frac{{{\tau _1}}}{{{\Omega _1}}}}}} \right)\left( {1 - {e^{ - \frac{{{\tau _1}}}{{{\Omega _0}}}}}} \right)} \right] ,
\end{align}
where $\varphi  = \frac{1}{{\rho {a_1}{\Omega _0}}} - \frac{{{\gamma _{t{h_2}}^{FD}}}}{{\rho {\Omega _2}}} - {\phi _1}$, ${\phi _1} = \frac{{ - {a_2}}}{{{a_1}\rho {\Omega _2}}}$, ${\phi _2} = {a_1}\rho {\Omega _0}$, $\psi  = \frac{{ - {a_2}}}{{{a_1}\rho {\Omega _2}\left( {1 + {a_1}\rho \tau_{1} } \right)}}$ and $\chi  = \frac{{{\Omega _1}}}{{{\Omega _1} + \tau_{1} \rho {\Omega _{{\rm{LI}}}}}}$.  $\mathrm{Ei\left(\cdot\right)}$ is the exponential integral function~\cite[Eq. (8.211.1)]{gradshteyn}.
\end{theorem}
\begin{proof}
See Appendix~A.
\end{proof}

\subsubsection{Diversity Analysis}
In this subsection, the diversity order of $D_{2}$ with direct link for FD NOMA is analyzed in the following.
\paragraph{$D_{2}$ for FD NOMA case}
For $D_{2}$ with direct link, it is challenging to obtain diversity order from \eqref{OP derived for D2}. We can use Gaussian-Chebyshev quadrature to find an approximation from \eqref{express OP for D2} and the approximated expression of outage probability for $D_{2}$ at high SNR is given by
\begin{align}\label{appro OP for user2}
 &P_{{D_2},dir}^{FD,appro} = \left[ {\frac{{{\tau _1}}}{{{\Omega _0}}} - \left( {1 - \frac{{{\Omega _2}{\tau _1} + 2{\Omega _0}{\tau _1}\left( {{a_2} - {a_1}\gamma _{t{h_2}}^{FD}} \right)}}{{2{\Omega _0}{\Omega _2}}}} \right)} \right. \nonumber\\
  &\times \frac{{{\tau _1}\pi }}{{2N{\Omega _0}}}\sum\limits_{n = 1}^N {\left( {1 + \frac{{\left( {{s_n} + 1} \right){\tau _1}{a _2}}}{{{\Omega _2}\left( {\left( {{s_n} + 1} \right){\tau _1}{a _1}\rho  + 2} \right)}} - \frac{{{s_n}{\tau _1}}}{{2{\Omega _0}}}} \right)} \nonumber \\
  &\times \left. {\sqrt {1 - s_n^2} } \right]\frac{{{\Omega _1}}}{{\left( {{\Omega _1} + {\tau _1}\rho {\Omega _{{\rm{LI}}}}} \right)}}{\rm{ + }}\left( {1 - \frac{{{\Omega _1}}}{{\left( {{\Omega _1} + {\tau _1}\rho {\Omega _{{\rm{LI}}}}} \right)}}} \right)\frac{{{\tau _1}}}{{{\Omega _0}}} ,
\end{align}
where $N$ is a parameter to ensure a complexity-accuracy tradeoff, ${s_n} = \cos \left( {\frac{{2n - 1}}{{2N}}\pi } \right)$.
Substituting \eqref{appro OP for user2} into \eqref{diversity order}, we can obtain $d_{{D_2},dir}^{FD}  = 1$.
\begin{remark}
From above explanation, the observation is that the direct link $(BS \rightarrow D_{2})$ to convey information is an effective way to overcome the problem of zero diversity order for $D_{2}$.
\end{remark}
\paragraph{$D_{2}$ for HD NOMA case}
The outage performance of $D_{2}$ for HD NOMA has been investigated in \cite{Ding2014Cooperative} and we can obtain $d_{{D_2},dir}^{HD} = 2$.
%The diversity order of $D_{2}$ with direct link is two, i.e. $d_{{D_2},dir}^{HD} = 2$, which can be obtained in \cite{Ding2014Cooperative}.

\subsubsection{Throughput Analysis}\label{SecitonIII-B-3}
Based on the derived results of outage probability above, we obtain the throughput expressions for FD/HD NOMA in delay-limited transmission mode as below.
\paragraph{FD NOMA case}
As suggested in Section \ref{SecitonIII-A-4}, the system throughput of FD NOMA with direct link is given by
\begin{align}\label{Delay-limited Transmission direct for FD}
R_{l\_dir}^{FD}= \left( {1 - {P_{{D_1}}^{FD}}} \right){R_{{1}}} + \left( {1 - P_{{D_2},dir}^{FD}} \right){R_{{2}}},
\end{align}
where ${P_{{D_1}}^{FD}}$ and $P_{{D_2},dir}^{FD}$ can be obtained from \eqref{OP derived for D1 without directlink} and \eqref{OP derived for D2}, respectively.
\paragraph{HD NOMA case}
Similar to \eqref{Delay-limited Transmission direct for FD}, the system throughput of HD NOMA with direct link is given by
\begin{align}\label{Delay-limited Transmission direct for HD}
R_{l\_dir}^{HD} = \left( {1 - P_{{D_1}}^{HD}} \right){R_1} + \left( {1 - P_{{D_2},dir}^{HD}} \right){R_2},
\end{align}
where ${P_{{D_1}}^{HD}}$ and $P_{{D_2},dir}^{HD}$ can be obtained from \eqref{corollary1 derived for D1 no directlink for HD NOMA } and \cite[Eq. (11)]{Ding2014Cooperative}.

%\subsection{Ergodic Rate for Full-Duplex Cooperative NOMA}\label{SecitonIII-A-2}
\section{Ergodic rate}\label{ergodic rate}
%In this section, the performance of FD/HD user relaying are characterized in terms of ergodic sum rate in the following.
When user's rates are determined by their channel conditions, the ergodic sum rate is an important metric for performance evaluation. Hence the performance of FD/HD user relaying are characterized in terms of ergodic sum rates in the following.
\subsection{User Relaying without Direct Link}
\subsubsection{Ergodic Rate of $D_{1}$}
On the condition that $D_{1}$ can detect $x_{2}$, the achievable rate of $D_{1}$ can be written as ${R_{{D_1}}} = \log \left( {1 + \gamma_{D_{1}} } \right)$. The ergodic rate of $D_{1}$ for FD NOMA can be obtained in the following theorem.
\begin{theorem}\label{theorem_ergodic_rate_D1:4}
The closed-form expression of ergodic rate for $D_{1}$ without direct link for FD NOMA  is given by
\begin{align}\label{ergodic rate of D1 without directlink for FD}
%R_{{D_1}}^{FD} =& \mathbb{E}\left[ {\log \left( {1 + {{\left| {{h_1}} \right|}^2}{a_1}\rho } \right)} \right] \nonumber\\
 % =& \frac{{{a_1}\rho }}{{\ln 2}}\int_0^\infty  {\frac{{1 - {F_{{{\left| {{h_1}} \right|}^2}}}\left( x \right)}}{{1 + {a _1}\rho x}}} dx \nonumber \\
 % =&  \frac{{ - {e^{\frac{1}{{{a _1}\rho {\Omega _1}}}}}}}{{\ln 2}}{\mathop{\rm Ei}\nolimits} \left( {\frac{{ - 1}}{{{a _1}\rho {\Omega _1}}}} \right).
 R_{{D_1}}^{FD} =& \frac{{{a_1}{\Omega _1}}}{{\ln 2\left( {{\Omega _{LI}} - {a_1}{\Omega _1}} \right)}}\left[ {{e^{\frac{1}{{{a_1}\rho {\Omega _1}}}}}{{\mathop{\rm E}\nolimits} _{\mathop{\rm i}\nolimits} }\left( {\frac{{ - 1}}{{{a_1}\rho {\Omega _1}}}} \right)} \right. \nonumber \\
 & \left. { - {e^{\frac{1}{{\rho {\Omega _{LI}}}}}}{{\mathop{\rm E}\nolimits} _{\mathop{\rm i}\nolimits} }\left( {\frac{{ - 1}}{{\rho {\Omega _{LI}}}}} \right)} \right] .
\end{align}
\end{theorem}
\begin{proof} See Appendix B.
\end{proof}

As such, we can derive the ergodic rate of $D_{1}$ for HD NOMA in the following corollary.
\begin{corollary}\label{ergodic rate of D1 without directlink for HD}
The ergodic rate of $D_{1}$ for HD NOMA is given by
\begin{align}\label{ergodic rate for D1 with HD}
R_{{D_1}}^{HD} = \frac{{ - {e^{\frac{1}{{{a _1}\rho {\Omega _1}}}}}}}{{2\ln 2}}{\mathop{\rm Ei}\nolimits} \left( {\frac{{ - 1}}{{{a _1}\rho {\Omega _1}}}} \right).
\end{align}
\end{corollary}
\subsubsection{Ergodic Rate of $D_{2}$}
Since $x_{2}$ should be detected at $D_{2}$ as well as at $D_{1}$ for SIC, the achievable rate of $D_{2}$ without direct link for FD NOMA is written as $R_{{D_2}} = \log \left( {1 + \min \left( {{\gamma _{{D_2} \to {D_1}}}, {\gamma _{2,{D_2}}}} \right)} \right)$. The corresponding ergodic rate is given by
\begin{align}\label{ergodic rate of D2 without directlink for FD}
R_{{D_2},nodir}^{FD} = \frac{1}{{\ln 2}}\int_0^\infty  {\frac{{1 - {F_{{X_1}}}\left( {{x_1}} \right)}}{{1 + {x_1}}}} d{x_1},
\end{align}
where ${X_1} = \min \left( {\frac{{{{\left| {{h_1}} \right|}^2}{a _2}\rho }}{{{{\left| {{h_1}} \right|}^2}{a _1}\rho  + \varpi {{\left| {{h_{LI}}} \right|}^2}\rho  + 1}},{{\left| {{h_2}} \right|}^2}\rho } \right)$ with $\varpi  = 1$.
Obviously, it is difficult to obtain the CDF of $X_{1}$. However, in order to derive an accurate closed-form expression for the ergodic rate applicable to high SNR region, the following theorem provides the high SNR approximation.
\begin{theorem}\label{theorem_ergodic_rate_D2:4}
The asymptotic expression for ergodic rate of $D_{2}$ without direct link for FD NOMA in the high SNR region is given by
\begin{align}\label{appro ergodic rate of D2 without directlink for FD}
R_{{D_2},nodir}^{FD,\infty } &= \frac{1}{{\ln 2}}\left\{ {{e^{\frac{1}{{\rho {\Omega _2}}}}}\left[ {{\mathop{\rm Ei}\nolimits} \left( {\frac{{ - 1}}{{\rho {a _1}{\Omega _2}}}} \right) - {\mathop{\rm Ei}\nolimits} \left( {\frac{{ - 1}}{{\rho {\Omega _2}}}} \right)} \right]} \right. \nonumber \\
  &\times \left( {\frac{{{\Omega _1}}}{{{a _2}{\Omega _1} - \xi }}} \right) - \frac{{{e^{\frac{{{a _2}{\Omega _1}}}{{\rho {\Omega _2}\xi }}}}}}{\xi }\left[ {{\mathop{\rm Ei}\nolimits} \left( { - \frac{{{a _2}\xi  + {a _1}{a _2}{\Omega _1}}}{{\rho {a _1}\xi {\Omega _2}}}} \right)} \right. \nonumber \\
 &\left. { - \left. {{\mathop{\rm Ei}\nolimits} \left( { - \frac{{{a _2}{\Omega _1}}}{{\rho {\Omega _2}\xi }}} \right)} \right]\left( {\frac{{{a _1}{a _2}\Omega _1^2 + {a _2}{\Omega _1}\xi }}{{{a _2}{\Omega _1} - \xi }}} \right)} \right\} ,
\end{align}
where $\xi  = \left( {{\Omega _{LI}} - {a _1}{\Omega _1}} \right)$.
\end{theorem}
\begin{proof} See Appendix C.
\end{proof}

For $\varpi  = 0$, the ergodic rate of $D_{2}$ without direct link for HD NOMA is given by
 \begin{align}\label{ergodic rate of D2 without directlink for HD}
R_{{D_2},nodir}^{HD} = \frac{1}{{2\ln 2}}\int_0^{\frac{{{a _2}}}{{{a _1}}}} {\frac{{{e^{ - \frac{y}{{\rho \left( {{a _2} - y{a _1}} \right){\Omega _1}}} - \frac{y}{{\rho {\Omega _2}}}}}}}{{1 + y}}} dy.
 \end{align}
 As can be seen from the above expression, \eqref{ergodic rate of D2 without directlink for HD} does not have a closed-form solution. Corollary 4 gives the high SNR approximation.

\begin{corollary}\label{high_SNR_sum_ergodic_rate_nodirectlink:4}
The asymptotic expression for ergodic rate of $D_{2}$ without direct link for HD NOMA in the high SNR region is given by
\begin{align}\label{appro ergodic rate of D2 without directlink for HD}
R_{{D_2},nodir}^{HD,\infty }=\frac{{{e^{\frac{1}{{\rho {\Omega _2}}}}}}}{{2\ln 2}}\left[ {{\mathop{\rm Ei}\nolimits} \left( {\frac{{ - 1}}{{\rho {a _1}{\Omega _2}}}} \right) - {\mathop{\rm Ei}\nolimits} \left( {\frac{{ - 1}}{{\rho {\Omega _2}}}} \right)} \right].
\end{align}
\begin{proof} See Appendix D.
\end{proof}
\end{corollary}

\subsubsection{Slope Analysis}
In this subsection, the high SNR slope is evaluated, which is the key parameter determining  ergodic rate in high SNR region. The high SNR slope is defined as
\begin{align}\label{high SNR slop}
S = \mathop {\lim }\limits_{\rho  \to \infty } \frac{{R_{D}^\infty \left( \rho  \right)}}{{\log \left( \rho  \right)}}.
\end{align}
%using ${\mathop{\rm Ei}\nolimits} \left( x \right)\mathop  \approx \limits^{x \to 0} \ln \left( { - x} \right) + x + C,x < 0$, where $C$ is Euler's constant.

\paragraph{$D_{1}$ for FD NOMA case}
Based on \eqref{ergodic rate of D1 without directlink for FD}, when $ \rho\ $$\rightarrow$ $\infty$, by using ${\mathop{\rm Ei}\nolimits} \left( { - x} \right) \approx \ln \left( x \right) + C$~\cite[Eq. (8.212.1)]{gradshteyn} and ${e^{-x}}\approx 1- x$, where $C$ is the Euler constant, the asymptotic ergodic rate of $D_{1}$ for FD NOMA is given by
\begin{align}\label{asymp rate of user1 for FD without directlink}
 R_{{D_1}}^{FD,\infty } =& \frac{{{a_1}{\Omega _1}}}{{\ln 2\left( {{\Omega _{LI}} - {a_1}{\Omega _1}} \right)}}\left[ {\left( {1 + \frac{1}{{{a_1}\rho {\Omega _1}}}} \right)} \right.\left( {\ln \left( {\frac{1}{{{a_1}\rho {\Omega _1}}}} \right)} \right. \nonumber\\
 &\left. { + C} \right)\left. { - \left( {1 + \frac{1}{{\rho {\Omega _{LI}}}}} \right)\left( {\ln \left( {\frac{1}{{\rho {\Omega _{LI}}}}} \right) + C} \right)} \right].
\end{align}
Substituting \eqref{asymp rate of user1 for FD without directlink} into \eqref{high SNR slop}, we can obtain $S_{{D_1}}^{FD}  = 0$.
%the high SNR slop of $D_{1}$ for FD NOMA without direct link is zero, i.e. ${S_{{D_1},nodir}} = 1$.
\paragraph{$D_{1}$ for HD NOMA case}
Based on \eqref{ergodic rate for D1 with HD}, the asymptotic ergodic rate of $D_{1}$ for HD NOMA in the high SNR region is given by
\begin{align}\label{asymp rate of user1 for HD without directlink}
R_{{D_1}}^{HD,\infty } = \frac{{ - 1}}{{2\ln 2}}\left( {1 + \frac{1}{{{a_1}\rho {\Omega _1}}}} \right)\left[ {\ln \left( {\frac{1}{{{a_1}\rho {\Omega _1}}}} \right) + C} \right].
\end{align}
Substituting \eqref{asymp rate of user1 for HD without directlink} into \eqref{high SNR slop}, we can obtain $S_{{D_1}}^{HD} = \frac{1}{2}$.
%the high SNR slop of $D_{2}$ for HD NOMA without direct link is $\frac{1}{2}$, i.e, $S_{{D_1},nodir}^{HD} = \frac{1}{2}$.
\paragraph{$D_{2}$ for FD NOMA case}\label{User2 for FD NOMA case slop without directlink}
Based on above analysis, substituting \eqref{appro ergodic rate of D2 without directlink for FD} into \eqref{high SNR slop}, we can obtain $S_{{D_2},nodir}^{FD}  =  0$.
\paragraph{$D_{2}$ for HD NOMA case}
Such as \eqref{User2 for FD NOMA case slop without directlink}, substituting \eqref{appro ergodic rate of D2 without directlink for HD} into \eqref{high SNR slop}, we can obtain $S_{{D_2},nodir}^{HD} =0 $.
\begin{remark}\label{ceiling for D2 without direct link}
Based on above analysis, the ergodic rate of $D_{2}$ converges to a throughput ceiling in the high SNR region for FD/HD NOMA without direct link.
\end{remark}
Combing \eqref{appro ergodic rate of D2 without directlink for FD} and \eqref{asymp rate of user1 for FD without directlink}, the asymptotic expression for ergodic sum rate of FD NOMA without direct link is expressed as
\begin{align}\label{ergodic sum rate without directlink for FD NOMA}
 &R_{sum,nodir}^{FD,\infty } = \frac{{{a_1}{\Omega _1}}}{{\ln 2\left( {{\Omega _{LI}} - {a_1}{\Omega _1}} \right)}}\left[ {\left( {1 + \frac{1}{{{a_1}\rho {\Omega _1}}}} \right)} \right. \nonumber \\
 & \times \left( {\ln \left( {\frac{1}{{{a_1}\rho {\Omega _1}}}} \right) + C} \right)\left. { - \left( {1 + \frac{1}{{\rho {\Omega _{LI}}}}} \right)\left( {\ln \left( {\frac{1}{{\rho {\Omega _{LI}}}}} \right) + C} \right)} \right] \nonumber\\
  &+ \frac{1}{{\ln 2}}\left\{ {{e^{\frac{1}{{\rho {\Omega _2}}}}}} \right.\left[ {{\mathop{\rm Ei}\nolimits} \left( {\frac{{ - 1}}{{\rho {a_1}{\Omega _2}}}} \right) - {\mathop{\rm Ei}\nolimits} \left( {\frac{{ - 1}}{{\rho {\Omega _2}}}} \right)} \right]\left( {\frac{{{\Omega _1}}}{{{a_2}{\Omega _1} - \xi }}} \right) \nonumber\\
  &+ \frac{{{e^{\frac{{{a_2}{\Omega _1}}}{{\rho {\Omega _2}\xi }}}}}}{\xi }\left[ {{\mathop{\rm Ei}\nolimits} \left( {\frac{{ - {a_2}{\Omega _1}}}{{\rho {\Omega _2}\xi }}} \right) - {\mathop{\rm Ei}\nolimits} \left( {\frac{{ - {a_2}\xi  - {a_1}{a_2}{\Omega _1}}}{{\rho {a_1}\xi {\Omega _2}}}} \right)} \right] \nonumber\\
 & \left. { \times \left( {\frac{{{a_1}{a_2}\Omega _1^2 + {a_2}{\Omega _1}\xi }}{{{a_2}{\Omega _1} - \xi }}} \right)} \right\} .
\end{align}

Similarly, combing \eqref{appro ergodic rate of D2 without directlink for HD} and \eqref{asymp rate of user1 for HD without directlink}, the asymptotic expression for ergodic sum rate of HD NOMA without direct link is expressed as
\begin{align}\label{ergodic sum rate without directlink for HD NOMA}
 R_{sum,nodir}^{HD,\infty } =& \frac{{ - 1}}{{2\ln 2}}\left( {1 + \frac{1}{{{a_1}\rho {\Omega _1}}}} \right)\left[ {\ln \left( {\frac{1}{{{a_1}\rho {\Omega _1}}}} \right) + C} \right] \nonumber \\
&  + \frac{{{e^{\frac{1}{{\rho {\Omega _2}}}}}}}{{2\ln 2}}\left[ {{\mathop{\rm Ei}\nolimits} \left( {\frac{{ - 1}}{{\rho {a_1}{\Omega _2}}}} \right) - {\mathop{\rm Ei}\nolimits} \left( {\frac{{ - 1}}{{\rho {\Omega _2}}}} \right)} \right]   .
\end{align}
\subsubsection{Throughput Analysis}\label{Seciton-A-3}
In this subsection, the throughput in delay-tolerant transmission for FD/HD NOMA are presented, respectively.
\paragraph{FD NOMA case}
In this mode, the throughput is determined by evaluating the ergodic rate. Using \eqref{ergodic rate of D1 without directlink for FD} and \eqref{ergodic rate of D2 without directlink for FD}, the system throughput of FD NOMA without direct link is given by
\begin{align}\label{Delay-tolerant Transmission nodirect for FD}
R_{t\_nodir}^{FD} = R_{{D_1}}^{FD} +R_{{D_2},nodir}^{FD}.
\end{align}
\paragraph{HD NOMA case}
Similar to \eqref{Delay-tolerant Transmission nodirect for FD}, using \eqref{ergodic rate for D1 with HD} and \eqref{ergodic rate of D2 without directlink for HD}, the system throughput of HD NOMA without direct link is given by
\begin{align}\label{Delay-tolerant Transmission nodirect for HD}
R_{t\_nodir}^{HD} = R_{{D_1}}^{HD} + R_{{D_2},nodir}^{HD}.
\end{align}

\subsection{User Relaying with Direct Link}
In this subsection, we investigate the ergodic rate of $D_{2}$ for FD/HD NOMA with direct link.
\subsubsection{Ergodic Rate of $D_{2}$}
Assume that the signal $x_{2}$ from relaying and direct link can be detected at ${D_{2}}$ as well as at $D_{1}$ for SIC. Moreover, considering the effect of RI between these two links, the achievable rate of $D_{2}$ is written as $R_{{D_2},dir}^{RI} = \log \left( {1 + \min \left( {{\gamma _{{D_2} \to {D_1}}},\gamma _{1,{D_2}}^{RI} + \gamma _{2,{D_2}}^{RI}} \right)} \right)$. For the sake of simplicity, the achievable rate for $D_{2}$ can be further written as $R_{{D_2},dir} = \log \left( {1 + \min \left( {{\gamma _{{D_2} \to {D_1}}},\gamma _{{D_2}}^{MRC}} \right)} \right)$. Hence, the ergodic rate of $D_{2}$ for FD NOMA is given by
\begin{align}\label{ergodic rate for D2}
R_{{D_2},dir}^{FD}  = \frac{1}{{\ln 2}}\int_0^\infty  {\frac{{1 - {F_{{X_2}}}\left( {{x_2}} \right)}}{{1 + {x_2}}}} d{x_2},
\end{align}
where ${X_2} = \min \left( {\frac{{{{\left| {{h_1}} \right|}^2}{a _2}\rho }}{{{{\left| {{h_1}} \right|}^2}{a _1}\rho  + \varpi {{\left| {{h_{LI}}} \right|}^2}\rho  + 1}},{{\left| {{h_2}} \right|}^2}\rho  + \frac{{{{\left| {{h_0}} \right|}^2}{a _2}\rho }}{{{{\left| {{h_0}} \right|}^2}{a _1}\rho  + 1}}} \right)$ with $\varpi =1$.
It is also difficult to obtain the CDF of $X_{2}$ and \eqref{ergodic rate for D2} has not closed-form expression. The following theorem provides the high SNR approximation.
\begin{theorem}\label{theorem:6 direct link rate for D2}
The asymptotic expression for ergodic rate of $D_{2}$ with direct link for FD NOMA in the high SNR region is given by
\begin{align}\label{appro ergodic rate of D2 with directlink for FD}
 R_{{D_2},dir}^{FD,\infty } =& \frac{1}{{\ln 2}}\left\{ {\ln \left( {1 + \frac{{{a _2}}}{{{a _1}}}} \right)\left( {\frac{{{\Omega _1}}}{{{a _2}{\Omega _1} - \xi }}} \right)} \right. \nonumber \\
& \left. { - \frac{1}{\xi }\ln \left( {1 + \frac{\xi }{{{a _1}{\Omega _1}}}} \right)\left( {\frac{{{a _2}{\Omega _1}\xi  + {a _1}{a _2}\Omega _1^2}}{{{a _2}{\Omega _1} - \xi }}} \right)} \right\} .
\end{align}
\begin{proof} See Appendix E.
\end{proof}
\end{theorem}

For $\varpi  = 0$, the ergodic rate of $D_{2}$ for HD NOMA with direct link is given by
\begin{align}\label{ergodic rate of D2 with directlink for HD}
 &R_{{D_2},dir}^{HD}  = \int_0^{\frac{{{a _2}}}{{{a _1}}}} {\frac{{{e^{ - \frac{{y\left( {{\Omega _0} + {\Omega _1}} \right)}}{{\rho \left( {{a _2} - {a _1}y} \right){\Omega _0}{\Omega _1}}}}}}}{{1 + y}}} dy \nonumber \\
  &+ \int_0^{\frac{{{a _2}}}{{{a _1}}}} {\int_0^{\frac{y}{{\rho \left( {{a _2} - {a _1}y} \right)}}} {\frac{{{e^{ - \frac{x}{{{\Omega _0}}} - \frac{{{\gamma _{th_2}^{HD}}\left( {x{a _1}\rho  + 1} \right) - x{a _2}\rho }}{{\rho \left( {x{a _1}\rho  + 1} \right){\Omega _2}}} - \frac{y}{{\rho \left( {{a _2} - {a _1}y} \right){\Omega _1}}}}}}}{{\left( {1 + y} \right){\Omega _0}}}} dx} dy .
\end{align}
To obtain the closed-form expression of ergodic rate for $D_{2}$, the complicated integrals are required to be computed. Corollary 5 gives an efficient high SNR approximation.
\begin{corollary}\label{high_SNR_sum_ergodic_rate_directlink:5}
The asymptotic expression for ergodic rate of $D_{2}$ with direct link for HD NOMA in the high SNR region is given by
\begin{align}\label{appro ergodic rate of D2 with direct link for HD NOMA}
 R_{{D_2},dir}^{HD,\infty } = \frac{{\rm{1}}}{{\rm{2}}}\log \left( {1 + \frac{{{a _2}}}{{{a _1}}}} \right).
\end{align}
\end{corollary}

Combing \eqref{asymp rate of user1 for FD without directlink} and \eqref{appro ergodic rate of D2 with directlink for FD}, the asymptotic expression for the ergodic sum rate of FD NOMA with direct link is expressed as
\begin{align}\label{ergodic sum rate with directlink for FD NOMA}
& R_{sum,dir}^{FD,\infty } = \frac{{{a_1}{\Omega _1}}}{{\ln 2\left( {{\Omega _{LI}} - {a_1}{\Omega _1}} \right)}}\left[ {\left( {1 + \frac{1}{{{a_1}\rho {\Omega _1}}}} \right)} \right.\left( {\ln \left( {\frac{1}{{{a_1}\rho {\Omega _1}}}} \right)} \right. \nonumber\\
& \left. { + C} \right)\left. { - \left( {1 + \frac{1}{{\rho {\Omega _{LI}}}}} \right)\left( {\ln \left( {\frac{1}{{\rho {\Omega _{LI}}}}} \right) + C} \right)} \right] + \frac{1}{{\ln 2}} \nonumber\\
  &\times \left[ {\ln \left( {1 + \frac{{{a_2}}}{{{a_1}}}} \right)\left( {\frac{{{\Omega _1}}}{{{a_2}{\Omega _1} - \xi }}} \right) - \frac{1}{\xi }\ln \left( {1 + \frac{\xi }{{{a_1}{\Omega _1}}}} \right)} \right. \nonumber\\
 &\left. { \times \left( {\frac{{{a_2}{\Omega _1}\xi  + {a_1}{a_2}\Omega _1^2}}{{{a_2}{\Omega _1} - \xi }}} \right)} \right]
\end{align}

In the same way, combing \eqref{asymp rate of user1 for HD without directlink} and \eqref{appro ergodic rate of D2 with direct link for HD NOMA}, the asymptotic expression for ergodic sum rate of HD NOMA with direct link is expressed as
\begin{align}\label{ergodic sum rate with directlink for HD NOMA}
 R_{sum,dir}^{HD,\infty } = &\frac{{ - 1}}{{2\ln 2}}\left( {1 + \frac{1}{{{a_1}\rho {\Omega _1}}}} \right)\left[ {\ln \left( {\frac{1}{{{a_1}\rho {\Omega _1}}}} \right) + C} \right] \nonumber\\
  &+ \frac{{\rm{1}}}{{\rm{2}}}\log \left( {1 + \frac{{{a_2}}}{{{a_1}}}} \right) .
\end{align}

\subsubsection{Slope Analysis}
Based on the derived asymptotic ergodic rates, the high SNR slopes of $D_{2}$ with direct link are characterized in the following.
\paragraph{$D_{2}$ for FD NOMA case}\label{User2 for FD NOMA case slop with directlink}
Substituting \eqref{appro ergodic rate of D2 with directlink for FD} into \eqref{high SNR slop}, we can obtain $S_{{D_2},dir}^{FD}  = 0$.
%are zeros, i.e. ${S_{{D_2},dir}} = S_{{D_2},dir}^{HD} = 0$.
\paragraph{$D_{2}$ for HD NOMA case}
Such as \eqref{User2 for FD NOMA case slop with directlink}, substituting \eqref{appro ergodic rate of D2 with direct link for HD NOMA} into \eqref{high SNR slop}, we can obtain $ S_{{D_2},dir}^{HD} = 0$.
\begin{remark}\label{Remark 6}
Based on above derived results, the ergodic rate of $D_{2}$ also converge to a throughput ceiling in the high SNR region with direct link for FD/HD NOMA. The user of direct link is incapable of assisting $D_{2}$ to obtain the additional high SNR slope.
\end{remark}

%between the BS and $D_{2}$
\subsubsection{Throughput Analysis}\label{Seciton-IV-B-3}
\paragraph{FD NOMA case}
As suggested in Section \ref{Seciton-A-3}, in delay-tolerant transmission mode, using \eqref{ergodic rate of D1 without directlink for FD} and \eqref{ergodic rate for D2}, the system throughput for FD NOMA with direct link is given by
\begin{align}\label{Delay-tolerant Transmission direct for FD}
R_{t\_dir}^{FD} = R_{{D_1}}^{FD} +  R_{{D_2},dir}^{FD}.
\end{align}
\paragraph{HD NOMA case}
Similar to \eqref{Delay-tolerant Transmission direct for FD}, using \eqref{ergodic rate for D1 with HD} and \eqref{ergodic rate of D2 with directlink for HD}, the system throughput for HD NOMA with direct link is given by
\begin{align}\label{Delay-tolerant Transmission direct for HD}
R_{t\_dir}^{HD} = R_{{D_1}}^{HD} + R_{{D_2},dir}^{HD}.
\end{align}

As shown in TABLE I, the diversity orders and high SNR slopes of two users for FD/HD NOMA are summarized to illustrate the comparison between them. In TABLE I, we use ``D'' and ``S'' to represent the diversity order and high SNR slope, respectively.
\begin{table}[!h]
\begin{center}
{\tabcolsep12pt\begin{tabular}{|l|l|l|l|l|}\hline   %{p{3cm}p{3cmp}p{3cm}p{3cm}}%
 %% \diagbox[width=3em,trim=l]{}   &\diagbox[width=3em,trim=l]{}  &\diagbox[width=3em,trim=l]{} & D & S \\
  \textbf{Duplex mode} & \textbf{Link mode}  &\textbf{User} & \textbf{D} & \textbf{S} \\
     \hline
\multirow{4}{*}{FD NOMA}    & \multirow{2}{*}{Nodirect}   & User1 & $ 0$ & $  0$ \\
\cline{3-5}
                     &  & User2  & $ 0$ & $ 0$  \\
\cline{2-5}
                     & \multirow{2}{*}{Direct} & User1  & $0$   & $ 0$  \\
\cline{3-5}
               &  &  User2 & $ 1$   & $ 0$ \\
\hline
\multirow{4}{*}{HD NOMA}  & \multirow{2}{*}{Nodirect} & User1& $ 1$ & $ \frac{1}{2} $ \\
\cline{3-5}
               &  & User2& $ 1$ & $ 0   $ \\
\cline{2-5}
               &\multirow{2}{*}{Direct}  &User1 & $ 1$   & $\frac{1}{2}$  \\
\cline{3-5}
               &  &User2 & $ 2$   & $  0$ \\
\hline

\end{tabular}}{}
\label{tab1}
\end{center}
\caption{Diversity order and high SNR slope for FD/HD NOMA systems.}
\end{table}
%\begin{table}[!h]
%\begin{center}
%{\tabcolsep16pt\begin{tabular}{|l|l|l|} \hline   %{p{3cm}p{3cmp}p{3cm}p{3cm}}%
% & Diversity order & high SNR slop \\
%\hline
%FD NOMA   & ${d_{{D_1},nodir}} = 0$ & ${S_{{D_1},nodir}} =  1$\\
 %         & ${d_{{D_2},nodir}} = 0$ & ${S_{{D_2},nodir}} =  0$ \\
  %         & ${d_{{D_1},dir}} = 0$   & ${S_{{D_1},dir}} =  1$  \\
   %      & ${d_{{D_2},dir}} = 1$   & ${S_{{D_2},nodir}} =  0$ \\
%\hline
%HD NOMA & ${d_{{D_1},nodir}} = 1$ & ${S_{{D_1},nodir}} =  \frac{1}{2} $ \\
 %       & ${d_{{D_2},nodir}} = 1$ & ${S_{{D_2},nodir}} =  0   $ \\
  %      & ${d_{{D_1},dir}} = 1$   & ${S_{{D_1},dir}} = \frac{1}{2}$  \\
   %     & ${d_{{D_2},dir}} = 2$   & ${S_{{D_2},dir}} =  0$ \\
%\hline
%\end{tabular}}{}
%\label{tab1}
%\end{center}
%\caption{Diversity order and high SNR slop for FD/HD NOMA}
%\end{table}
%%%% test-----

\section{Energy Efficiency}\label{Energy Efficiency}
 Based on throughput analysis, we aim to provide the system energy efficiency (EE) considering user relaying for FD/HD NOMA systems.

 The definition of energy efficiency is given by
\begin{align}
{\eta _{EE}} = \frac{{{\rm{Total~data~rate}}}}{{{\rm{Total~energy~consumption}}}}.
\end{align}

For FD/HD NOMA energy efficiency, the total data rate is denotes as sum throughput from the BS to $D_{1}$ and $D_{2}$ and from $D_{1}$ to $D_{2}$. The total power consumption is denoted as the sum of the transmitted power $P_{s}$ at the BS and $P_{r}$ at $D_{1}$. Based on results in Section \ref{SecitonIII-A-4}, \ref{SecitonIII-B-3} and \ref{Seciton-A-3}, \ref{Seciton-IV-B-3}, the energy efficiency of user relaying for FD/HD NOMA systems are expressed as
\begin{align}\label{EE FD for limited}
{\eta _\Phi^{FD} } = \frac{{{R_\Phi^{FD} }}}{{T{P_s} + T{P_r}}},
\end{align}
and
\begin{align}\label{EE HD for limited}
\eta _\Phi ^{HD} = \frac{{2R_\Phi ^{HD}}}{{T{P_s} + T{P_r}}},
\end{align}
respectively. $T$ denotes the transmission time for the entire communication process. $\Phi  \in \left( {l\_nodir,l\_dir,t\_nodir,t\_dir} \right)$. ${\eta _{l\_nodir}}$ and ${\eta _{l\_dir}}$ are system energy efficiencies without/with direct link in delay-limited transmission mode, respectively. ${\eta _{t\_nodir}}$ and ${\eta _{t\_dir}}$ are system energy efficiencies without/with direct link in delay-tolerant transmission mode, respectively.

\section{Numerical Results}\label{Numerical Results}
In this section, simulation results are provided to validate our analytical expressions derived in the previous section, and further evaluate the performance of FD/HD user relaying in NOMA systems. Without loss of generality, we assume that the distance between BS and $D_{2}$ is normalized to unity, i.e. ${\Omega _{{0}}} = 1$. ${\Omega _{{1}}} = {d^{ - \alpha }}$ and ${\Omega _{2}} = {\left( {1 - d} \right)^{ - \alpha }}$, where $d$ is the normalised distance between BS and $D_{1}$
setting to be $d=0.3$ and $\alpha $ is the pathloss exponent setting to be $\alpha=2$. The power allocation coefficients of NOMA are ${a_1}=0.2$ and ${a _2}=0.8$ for $D_{1}$ and $D_{2}$, respectively.
\subsection{Without Direct Link}
For user relaying without direct link, the target rate is set to be ${R_1}=3$, ${R_2}=0.5$ bit per channel use (BPCU) for $D_{1}$ and $D_{2}$, respectively. The performance of conventional OMA is shown as a benchmark for comparison, in which the total communication process is finished in three slots. In the first slot, the BS sends information $x_{1}$ to $D_{1}$ and sends $x_{2}$ to $D_{1}$ in the second slot. In the last slot, $D_{1}$ decodes and forwards the information $x_{2}$ to $D_{2}$.
\subsubsection{Outage Probability}
Fig. \ref{10_nodirect_HD_FD_15DB} plots the outage probability of two users versus SNR without direct link and the value of LI is assumed to be $\mathbb{E}\{|h_{LI}|^2\}= -15$ dB. The exact theoretical curves for the outage probability of two users for FD/HD NOMA are plotted according to \eqref{OP derived for D1 without directlink}, \eqref{exect OP for D2 no directlink} and \eqref{corollary1 derived for D1 no directlink for HD NOMA }, \eqref{corollary1 derived for D2 no directlink for HD NOMA }, respectively. Obviously, the exact outage probability curves match precisely with the Monte Carlo simulation results. It is observed that the outage performance of FD NOMA exceeds HD NOMA and OMA on the condition of low SNR region. This is because LI is not dominant impact factor in the low SNR region for FD NOMA and answers the first question we raised in the introduction part. Especially, one can also observed that the outage behavior of $D_2$ for HD NOMA outperforms HD OMA ~\cite[Eq. (8)]{Kwon2010Optimal}.
The asymptotic outage probability curves of two users for HD NOMA are plotted according to \eqref{appro OP of user1 for HD without directlink}  and \eqref{appro OP of user2 for HD without directlink }, respectively. The asymptotic curves well approximate the exact performance curves in the high SNR region. It is shown that error floors exist in FD NOMA, which verify the conclusion in \textbf{Remark \ref{error floor existent for D1 and D2 without direct link}} and obtain zero diversity order. This is due to the fact that there is loop interference in FD NOMA. Another observation is that HD NOMA and OMA are superior to FD NOMA in the high SNR region. Therefore, we can select different operation mode for user relaying according to the different SNR levels in practical cooperative NOMA systems.
\begin{figure}[t!]
    \begin{center}
        \includegraphics[width=3.48in,  height=2.8in]{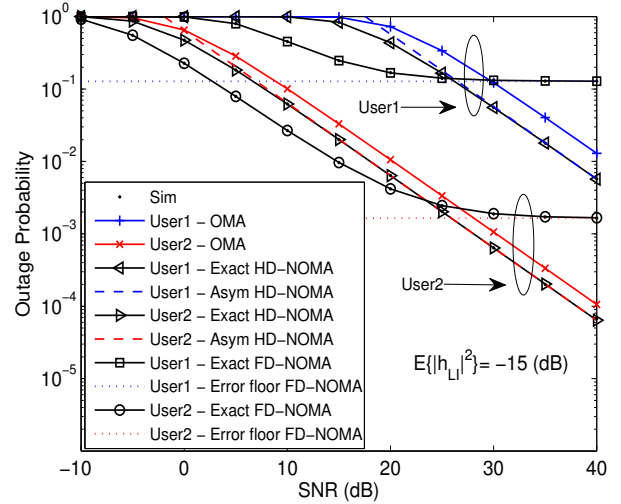}%   10_nodirect_HD_FD_5DB_new
        \caption{Outage probability versus the transmit SNR without direct link.}
        \label{10_nodirect_HD_FD_15DB}
    \end{center}
\end{figure}
\begin{figure}[t!]
    \begin{center}
        \includegraphics[width=3.48in,  height=2.8in]{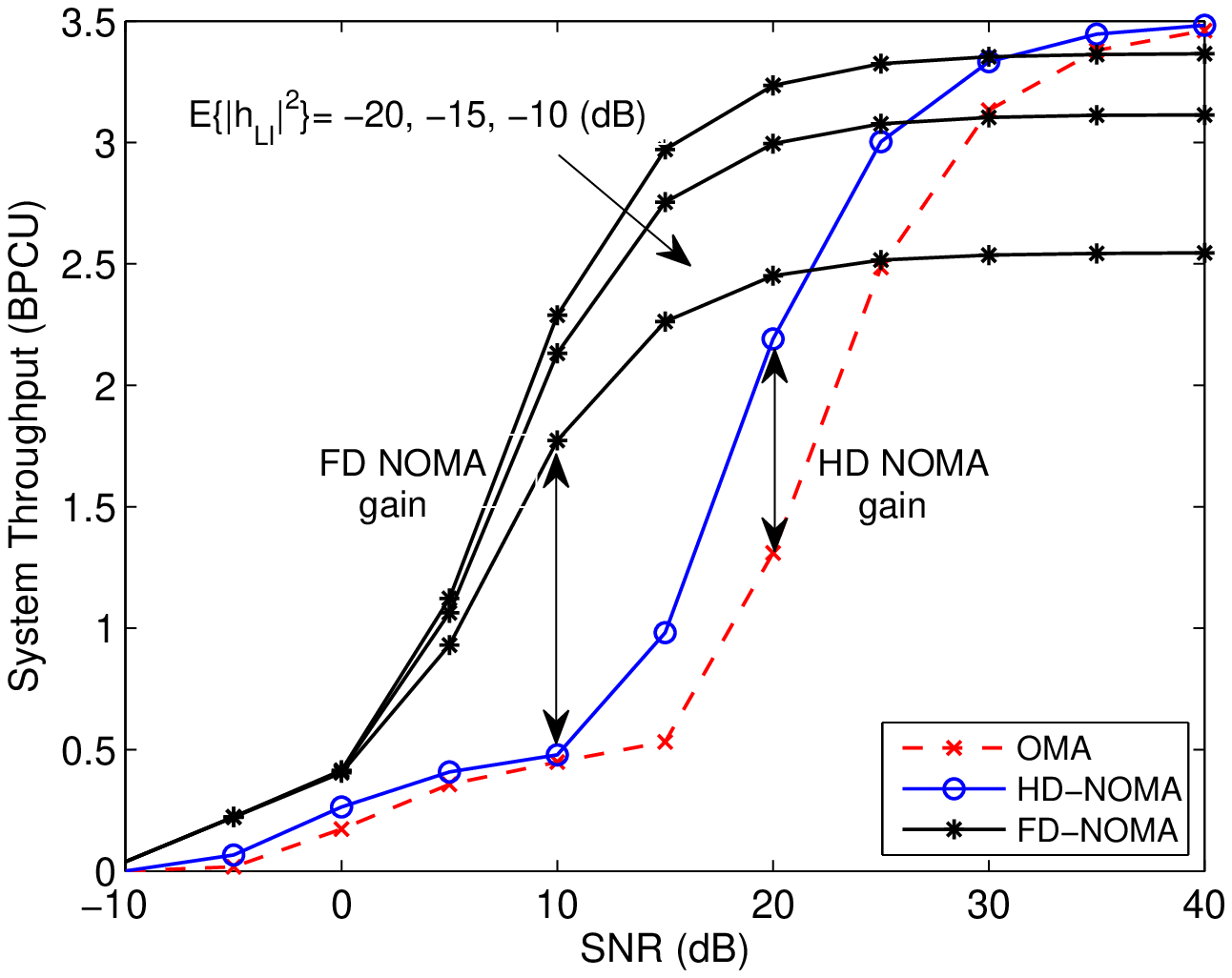}
        \caption{System throughput in delay-limited transmission mode versus SNR with different values of LI without direct link.}
        \label{nodirect_throughput_new}
    \end{center}
\end{figure}

\begin{figure}[t!]
    \begin{center}
        \includegraphics[width=3.48in,  height=2.8in]{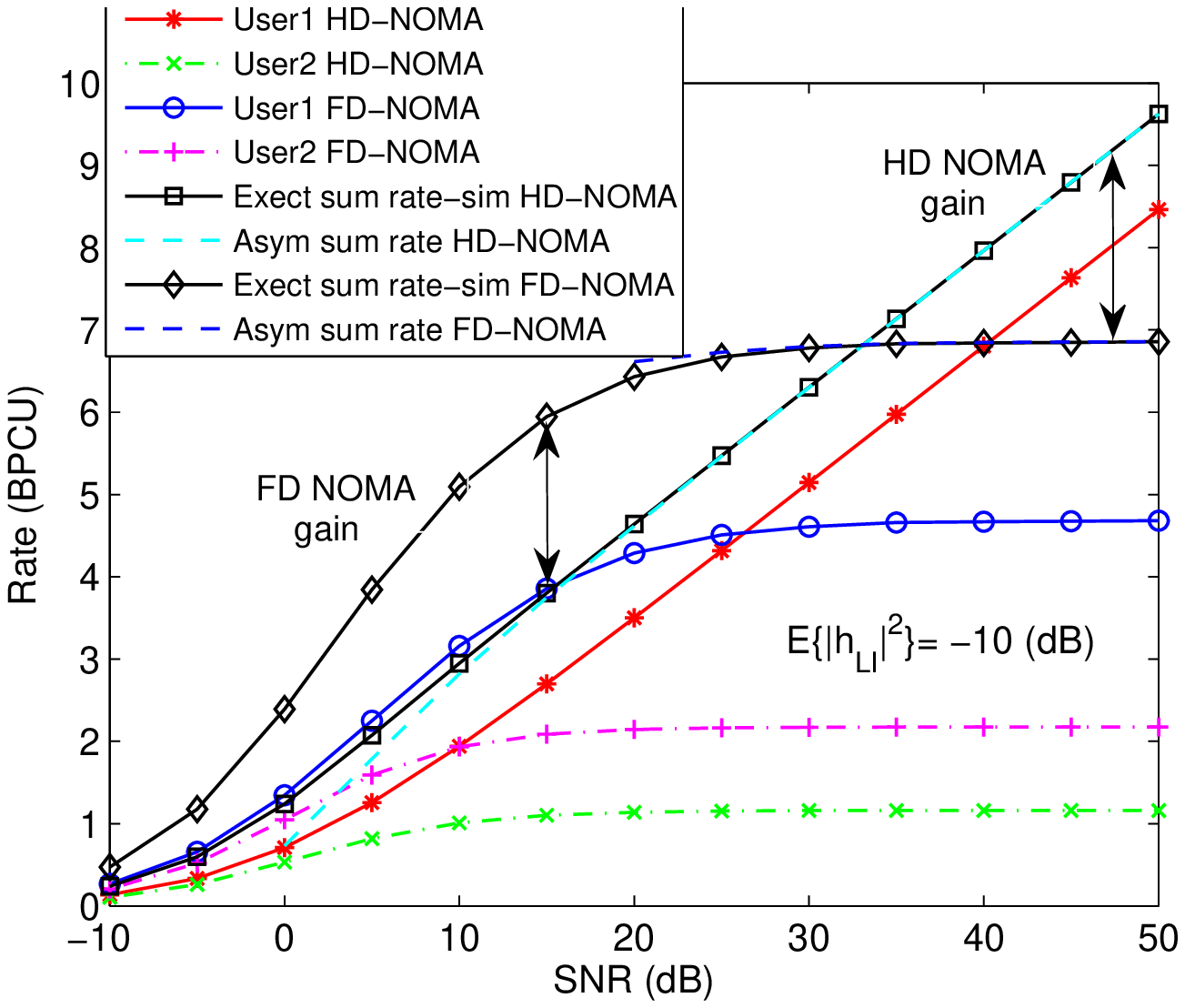} %nodirect_rate_HD_FD_new
        \caption{Rates versus the transmit SNR without direct link.}
        \label{nodirect_rate_HD_FD}
    \end{center}
\end{figure}
\begin{figure}[t!]
    \begin{center}
        \includegraphics[width=3.48in,  height=2.8in]{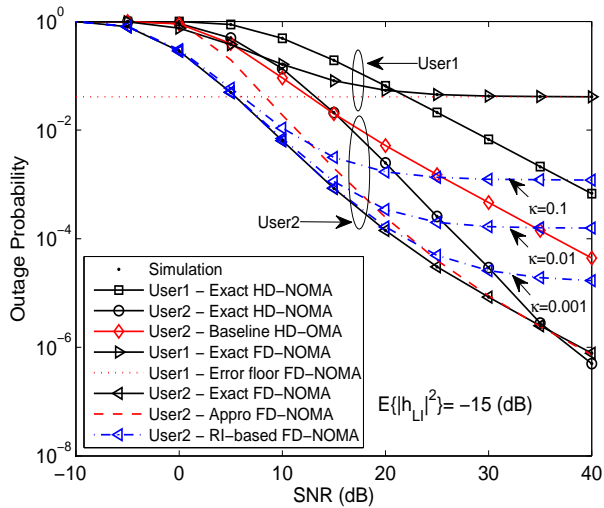} %OP_withdirect_HD_FD_add_RI
        \caption{Outage probability versus the transmit SNR with direct link.}
        \label{OP_withdirect_HD_FD_5}
    \end{center}
\end{figure}
%, where ${\mathop{\rm E}\nolimits} \{ {\left| {{h_{LI}}} \right|^2}\} = -5~dB$
Fig. \ref{nodirect_throughput_new} plots the system throughput versus SNR in delay-limited transmission mode without direct link. The solid curves represent throughput for FD/HD NOMA without direct link which are obtained from \eqref{Delay-limited Transmission nodirect for FD} and \eqref{Delay-limited Transmission nodirect for HD}, respectively.
We can observe that FD NOMA achieves a higher throughput than HD NOMA and OMA, since FD NOMA has the low values of LI. It is worth noting that increasing the values of LI from $-20$ dB to $-10$ dB reduce the system throughput in high SNR region. This is because FD NOMA converges to an error floor in high SNR region.
\subsubsection{Ergodic Rate}
Fig. \ref{nodirect_rate_HD_FD} plots the ergodic sum rate of FD/HD NOMA without direct link versus SNR and the value of LI is assumed to be $\mathbb{E}\{|h_{LI}|^2\}=-10$ dB. The blue and red solid curves denote the achievable rates of $D_{1}$ for FD/HD NOMA, respectively. The dashdotted curves denote the achievable rates of $D_{2}$ for FD/HD NOMA, respectively. One can observe that the achievable rate of $D_{1}$ for FD NOMA is superior to HD NOMA in the low SNR region. This phenomenon can be also explained is that LI has little effect on achievable rate of $D_{1}$ in the low SNR region. On the contrary, due to the influence of LI, the ergodic rate of $D_{1}$ converge to a throughput ceiling in the high SNR region. Another observation is that the achievable rate of $D_{2}$ for FD NOMA exceeds the HD NOMA. This is due to the fact that the communication process is completed over one slot time for FD NOMA. It is also shown that throughput ceilings exist in FD/HD NOMA for $D_{2}$, which verify the conclusion in \textbf{Remark \ref{ceiling for D2 without direct link}}. The dashed curves denote the asymptotic ergodic sum rate for FD/HD NOMA, corresponding to the analytical results derived in \eqref{ergodic sum rate without directlink for FD NOMA} and \eqref{ergodic sum rate without directlink for HD NOMA}, respectively. An important observation is that FD NOMA can achieve the maximal ergodic sum rate corresponding to HD NOMA and OMA in the low SNR region. The reason is that FD NOMA can improve system spectrum efficiency in the low SNR region. This phenomenon answers the third question we raised in the introduction part.

\subsection{With Direct Link}
For user relaying with direct link, the target rate is set to be ${R_1}=2$, ${R_2}=1$ BPCU for $D_{1}$ and $D_{2}$, respectively.  The performance of conventional HD NOMA is shown as a benchmark for comparison, which has been discussed in \cite{Ding2014Cooperative}.
\subsubsection{Outage Probability}
Fig. \ref{OP_withdirect_HD_FD_5} plots the outage probability of two users versus SNR and the value of LI is assumed to be $\mathbb{E}\{|h_{LI}|^2\}=-15$ dB. The exact outage probability curves of two users for FD NOMA are given by Monte Carlo simulations and perfectly match with the analytical results derived in \eqref{OP derived for D1 without directlink} and \eqref{OP derived for D2}. The approximated outage probability curve for $D_{2}$ is plotted according to \eqref{appro OP for user2} is practically indistinguishable from the exact expression. We observe that $D_{2}$ obtains one diversity order by using the direct link, which overcomes the problem of zero diversity order inherent to FD cooperative systems. This phenomenon answers the second question we raised in the introduction part. More importantly, one can observe that the performance of FD NOMA is superior to HD NOMA in the low SNR region, whilst the performance is inferior to HD NOMA in the high SNR region.
Additionally, for $D_2$, considering the impact of RI between relaying link and direct link, we plots the corresponding outage probability of $D_2$ based on \eqref{express OP for D2 adding RI} denoted by blue dash-dotted curves.  As can be seen from Fig. \ref{OP_withdirect_HD_FD_5}, these simulation results almost match with analytical result derived in \eqref{OP derived for D2} by utilizing the upper bound SINR in low SNR region. However, with the increase of RI levels $\kappa$, the RI-based simulation results for $D_2$ converge to a constant and provide an error floor in high SNR region. Hence, the effect of RI should be carefully addressed in practical FD NOMA systems. Another observation is that the outage behavior of $D_2$ for FD/HD NOMA outperforms HD OMA~\cite[Eq. (13)]{Bai7097002}. That is due to the fact that NOMA can provide more spectral efficiency compared to OMA.

Fig. \ref{OP_withdirect_HD_FD_LI} plots the outage probability of the two users versus different values of LI from $-20$ dB to $-10$ dB. We see that LI strongly affect the performance of FD NOMA systems. With the values of LI increasing, the superiority of FD NOMA is no longer apparent. Therefore, it is important to consider the influence of LI when designing practical FD NOMA systems.
\begin{figure}[t!]
    \begin{center}
        \includegraphics[width=3.48in,  height=2.8in]{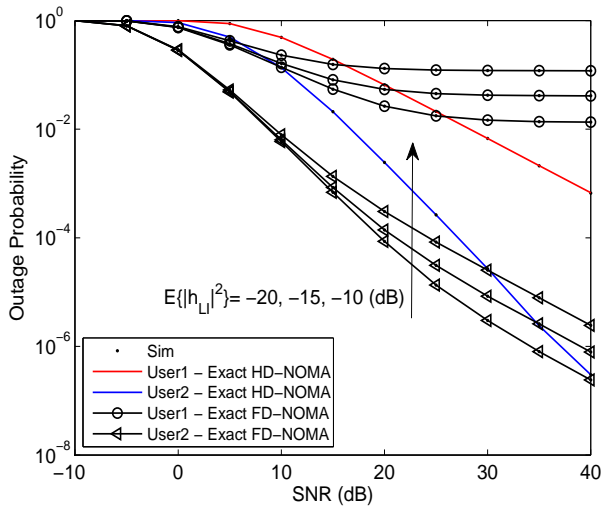}
        \caption{Outage probability versus the transmit SNR for different values of LI with direct link.}
        \label{OP_withdirect_HD_FD_LI}
    \end{center}
\end{figure}
Fig. \ref{FD_HD_NOMA_OP_throughput} plots system throughput versus SNR in delay-limited transmission mode with direct link. The solid curves, representing FD NOMA, is obtained from \eqref{Delay-limited Transmission direct for FD}. The dashed curve, representing HD NOMA, is obtained from \eqref{Delay-limited Transmission direct for HD}.
Observe that FD NOMA also outperform HD NOMA in the low SNR region. The reason is that in low SNR region, the outage probability is small and has no effect on the throughput, which only depends on the fixed transmission rates at the BS.
\begin{figure}[t!]
    \begin{center}
        \includegraphics[width=3.48in,  height=2.8in]{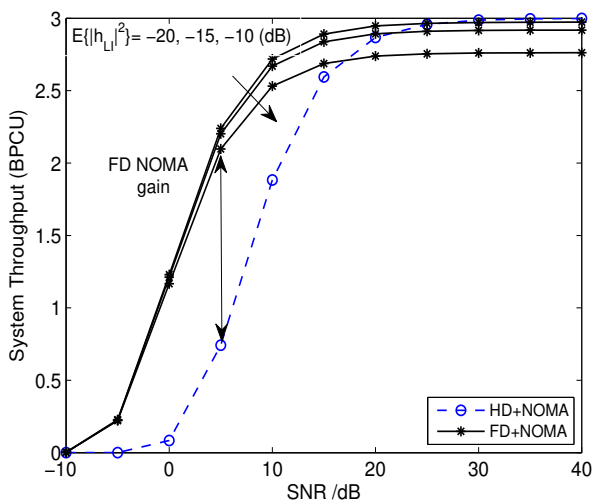}
       \caption{System throughput in delay-limited transmission mode versus SNR with different LI with direct link.}
        \label{FD_HD_NOMA_OP_throughput}
    \end{center}
\end{figure}
\begin{figure}[t!]
    \begin{center}
        \includegraphics[width=3.48in,  height=2.8in]{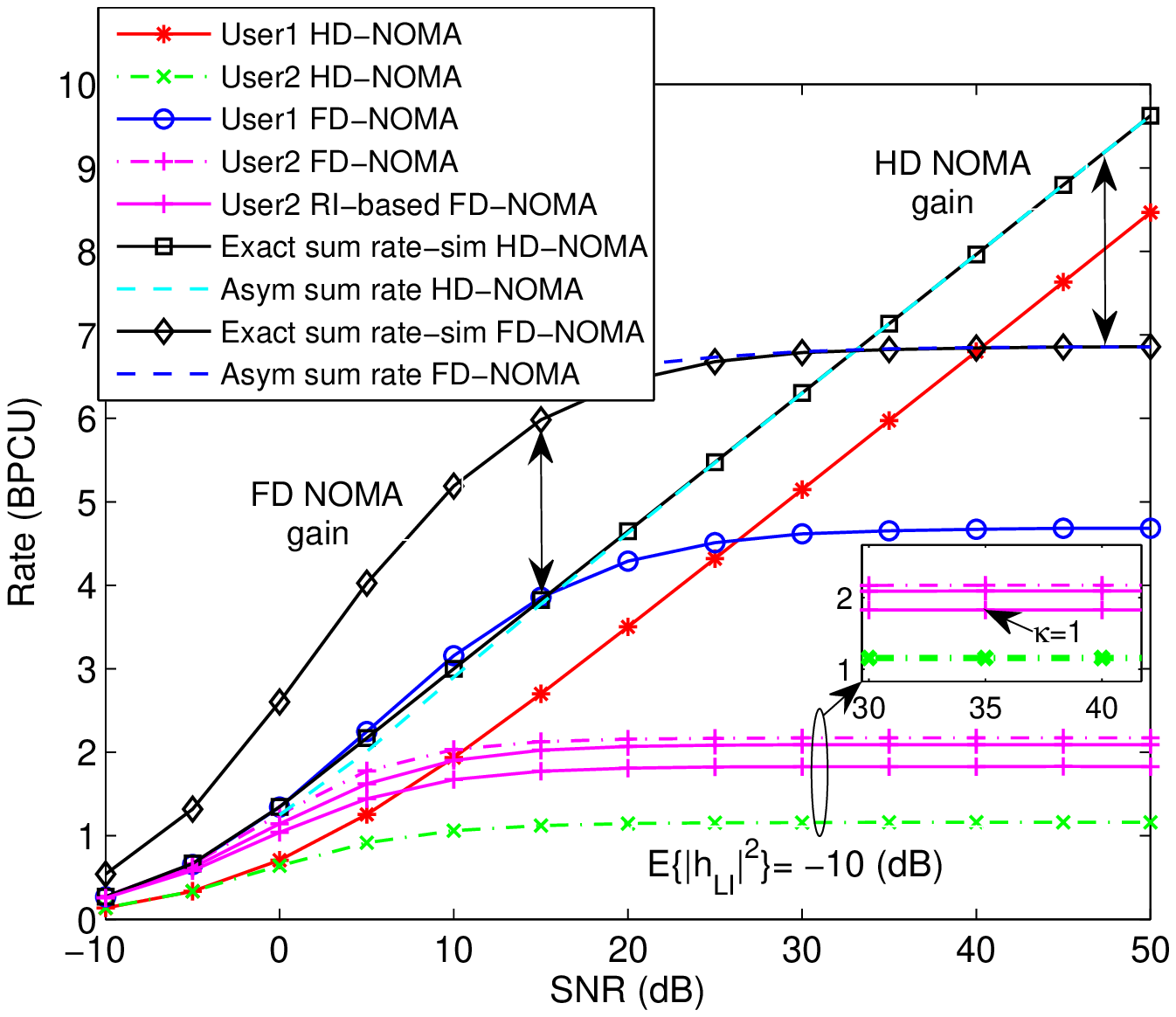}%direct_rate_HD_FD_new.eps
        \caption{Rates versus the transmit SNR with direct link.}
        \label{direct_rate_HD_FD}
    \end{center}
\end{figure}
\begin{figure}[t!]
    \begin{center}
        \includegraphics[width=3.48in,  height=2.8in]{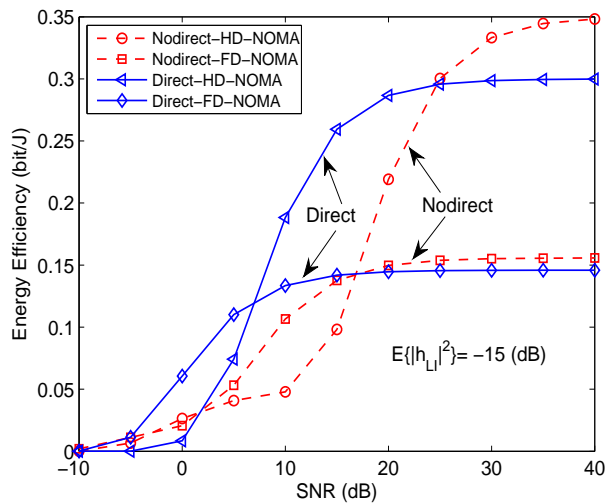}
        \caption{System energy efficiency in delay-limited transmission mode, where $P_{s}=P_{r}=$10 W, and $T=1$.}
        \label{EE_limited}
    \end{center}
\end{figure}
\begin{figure}[t!]
    \begin{center}
        \includegraphics[width=3.48in,  height=2.8in]{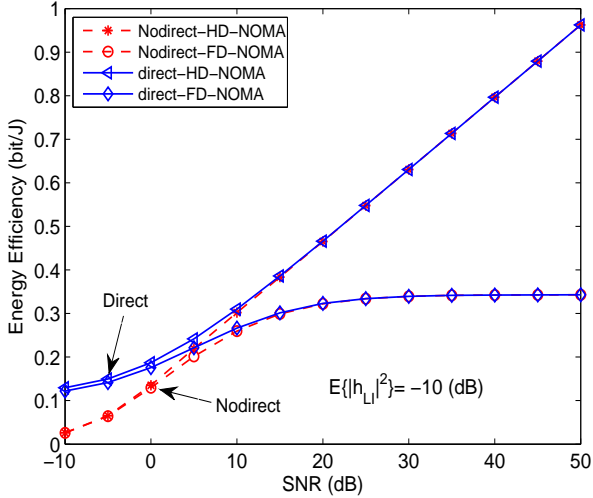}
        \caption{System energy efficiency in delay-tolerant transmission mode, where $P_{s}=P_{r}=$10 W, and $T=1$.}
        \label{EE_tolerant}
    \end{center}
\end{figure}

\subsubsection{Ergodic Rate}
Fig. \ref{direct_rate_HD_FD} plots the ergodic sum rate of HD/FD NOMA with direct link versus SNR and the value of LI is assumed to be $\mathbb{E}\{|h_{LI}|^2\}=-10$ dB. The dashed curves denote the asymptomatic ergodic sum rate for FD/HD NOMA based on the analytical results derived in \eqref{ergodic sum rate with directlink for FD NOMA} and \eqref{ergodic sum rate with directlink for HD NOMA}, respectively. It is observed that the asymptomatic ergodic sum rate is larger for FD/HD NOMA in the low SNR region. This can be explained as the direct link between BS and $D_{2}$ exists and improves system reliability.
One can observed from figure, as the RI value increases, the achievable rate of $D_2$ becomes smaller, such as, setting $\kappa$ from 0.5 to 1. In addition to consider the effect of LI, it is also important to design effective rake receiver at $D_2$ for FD NOMA system.

\subsection{Energy Efficiency}
Fig. \ref{EE_limited} plots the system energy efficiency versus SNR in delay-limited transmission mode for user relaying in NOMA systems. The dashed curves, representing user relaying without direct link for FD/HD NOMA are obtained from \eqref{EE FD for limited}, \eqref{Delay-limited Transmission nodirect for FD} and \eqref{EE HD for limited}, \eqref{Delay-limited Transmission nodirect for HD} with throughput in delay-limited transmission mode, respectively. The solid curves, representing representing user relaying with direct link for FD/HD NOMA are obtained from \eqref{EE FD for limited}, \eqref{Delay-limited Transmission direct for FD} and \eqref{EE HD for limited}, \eqref{Delay-limited Transmission direct for HD} with throughput in transmission mode, respectively. It can be seen that the energy efficiency of user relaying for FD/HD NOMA in delay-limited transmission mode is FD $>$ HD in the low SNR region. The reason is that FD NOMA can achieve larger throughput than that of HD NOMA at this transmission mode. This phenomenon answers the fourth question we raised in the introduction part.

Fig. \ref{EE_tolerant} plots the system energy efficiency versus SNR in delay-tolerant transmission mode for user relaying in NOMA systems. The dashed curves, representing user relaying without direct link for FD/HD NOMA are obtained from \eqref{ergodic sum rate without directlink for FD NOMA}, \eqref{EE FD for limited} and \eqref{ergodic sum rate without directlink for HD NOMA}, \eqref{EE HD for limited} with throughput in delay-tolerant mode, respectively. The solid curves, representing user relaying with direct link for FD/HD NOMA are obtained from \eqref{ergodic sum rate with directlink for FD NOMA}, \eqref{EE FD for limited} and \eqref{ergodic sum rate with directlink for HD NOMA}, \eqref{EE HD for limited} with throughput in delay-tolerant mode, respectively. We observe that
user relaying with direct link has a higher energy efficiency compared to without direct link for FD/HD NOMA in the low SNR region. This is because that the direct link improves system throughput at this transmission mode. Additionally, it is worth noting that HD NOMA achieves the higher system energy efficiency in the high SNR region. This is due to the fact that HD NOMA can provide a larger system throughput, while FD NOMA converges to the throughput ceiling in the high SNR region.

\section{Conclusion}\label{Conclusion}
This paper has investigated FD/HD user relaying in cooperative NOMA system and two cooperative relaying scenarios have been considered insightfully. The performance of FD/HD user relaying for NOMA system was characterized. The closed-form expressions of outage probability for two users have been derived. Due to the influence of LI, the diversity orders achieved by two user were zeros for FD NOMA. Therefore, the direct link between BS and far user was utilized to convey information and one diversity order was obtained for the far user. Based on the analytical results, it was shown that FD NOMA was superior to HD NOMA in low SNR region rather than in the high SNR region. The superior of FD NOMA was not apparent with the values of LI increasing. Furthermore, the expressions of ergodic sum rate for FD/HD user relaying were derived. The results showed that FD NOMA achieved a higher sum rate than HD NOMA in the low SNR region. In addition, the system energy efficiencies for FD/HD user relaying were discussed in different transmission modes.

\appendices
\section*{Appendix~A: Proof of Theorem \ref{theorem_outage:2}} \label{Appendix:A}
\renewcommand{\theequation}{A.\arabic{equation}}
\setcounter{equation}{0}

Based on \eqref{express OP for D2}, the outage probability of $D_{2}$ can be expressed as
\begin{align}\label{D2 derivd OP without direct}
 P_{{D_2},dir}^{FD} =& \underbrace {{{\rm{P}}_{\rm{r}}}\left( {\gamma _{{D_2}}^{MRC} < {\gamma _{t{h_2}}^{FD}}} \right)}_{{J_{11}}}\underbrace {{{\rm{P}}_{\rm{r}}}\left( {{\gamma _{{D_2} \to {D_1}}} > {\gamma _{t{h_2}}^{FD}}} \right)}_{{J_{12}}} \nonumber \\
  &+ \underbrace {{{\rm{P}}_{\rm{r}}}\left( {{\gamma _{1,{D_2}}} < {\gamma _{t{h_2}}^{FD}},{\gamma _{{D_2} \to {D_1}}} < {\gamma _{t{h_2}}^{FD}}} \right)}_{{J_{13}}}.
\end{align}
% and $J_{1}$ can be further written as follows:
%\begin{align}
%{J_1} =& \underbrace {{{\rm{P}}_{\rm{r}}}\left( {\gamma _{{D_2}}^{MRC} < {\gamma _{t{h_2}}^{FD}}} \right)}_{{J_{11}}}\underbrace {{{\rm{P}}_{\rm{r}}}\left( {{\gamma _{{D_2} \to {D_1}}} > {\gamma _{t{h_2}}^{FD}}} \right)}_{{J_{12}}}.
%\end{align}

Furthermore, substituting \eqref{SINR for D2 to detect D1}, \eqref{SINR for 1D2} and \eqref{MRC SINR for D2} to \eqref{D2 derivd OP without direct}, ${J_{11}}$ and ${J_{12}}$ can be calculated as
\begin{align}\label{J3_1}
 {J_{11}} =& {{\rm{P}}_{\rm{r}}}\left( {{{\left| {{h_2}} \right|}^2} < \frac{{{\gamma _{t{h_2}}^{FD}}}}{\rho } - \frac{{{{\left| {{h_0}} \right|}^2}{a _2}}}{{{{\left| {{h_0}} \right|}^2}{a _1}\rho  + 1}},{{\left| {{h_0}} \right|}^2} < \tau_{1} } \right)\nonumber \\
  =& \int_0^{\tau_{1}}  {{{\rm{F}}_{{{\left| {{h_2}} \right|}^2}}}\left( {\frac{{{\gamma _{t{h_2}}^{FD}}}}{\rho } - \frac{{y{a _2}}}{{y{a _1}\rho  + 1}}} \right)} {f_{{{\left| {{h_0}} \right|}^2}}}\left( y \right)dy \nonumber \\
  =& 1 - {e^{ - \frac{\tau_{1} }{{{\Omega _0}}}}} - \underbrace {\int_0^{\tau_{1}}  {\frac{1}{{{\Omega _0}}}{e^{ - \frac{y}{{{\Omega _0}}}}}{e^{ - \frac{1}{{{\Omega _2}}}\left( {\frac{{{\gamma _{t{h_2}}^{FD}}}}{\rho } - \frac{{y{a _2}}}{{y{a _1}\rho  + 1}}} \right)}}} dy}_{{\Theta _1}}.
\end{align}

Based on \eqref{J3_1}, using $x = y\rho {a _1} + 1$, ${\Theta _1}$ can be calculated as
\begin{align}\label{Theta1}
 {\Theta _1} =& \frac{1}{{{\Omega _0}}}{e^{ - \frac{{{\gamma _{t{h_2}}^{FD}}}}{{\rho {\Omega _2}}}}}\int_1^{\tau_{1} \rho {a _1} + 1} {{e^{ - \frac{{x - 1}}{{\rho {a _1}{\Omega _0}}}}}{e^{\frac{{{a _2}\left( {x - 1} \right)}}{{{a _1}\rho {\Omega _2}x}}}}} dx \nonumber \\
 % =& \frac{{{e^\varphi }}}{{{\phi _2}}}\int_1^{\tau \rho {\alpha _1} + 1} {{e^{ - \frac{x}{{\rho {\alpha _1}{\Omega _0}}}}}{e^{ - \frac{{{\alpha _2}}}{{{\alpha _1}\rho {\Omega _2}x}}}}} dx \nonumber \\
  =& \frac{{{e^\varphi }}}{{{\phi _2}}}\sum\limits_{n = 0}^\infty  {\frac{{{{\left( { - 1} \right)}^n}}}{{n!{{\left( {{\Omega _0}\rho {a _1}} \right)}^n}}}} \underbrace {\int_1^{\tau_{1} \rho {a _1} + 1} {{x^n}{e^{ - \frac{x}{{{a _1}\rho {\Omega _2}x}}}}} dx}_{{\Theta _2}},
  \end{align}
where $\varphi  = \frac{1}{{\rho {a _1}{\Omega _0}}} - \frac{{{\gamma _{t{h_2}}^{FD}}}}{{\rho {\Omega _2}}} - {\phi _1}$, ${\phi _1} = \frac{{ - {a _2}}}{{{a _1}\rho {\Omega _2}}}$ and ${\phi _2} = {a _1}\rho {\Omega _0}$. Note that \eqref{Theta1} is obtained by using Binomial theorem.

Furthermore, using $z = \frac{1}{x}$, ${{\Theta _2}}$ is given by
\begin{align} \label{Theta2}
 {\Theta _2} =& \int_{\frac{1}{{\tau_{1} \rho {a _1} + 1}}}^1 {\frac{1}{{{z^{n + 2}}}}{e^{ - \frac{{{a _2}z}}{{{a _1}\rho {\Omega _2}}}}}} dz \nonumber \\
 % =& \int_{\frac{1}{{\tau_{1} \rho {a _1} + 1}}}^\infty  {\frac{1}{{{z^{n + 2}}}}{e^{ - \frac{{{a _2}z}}{{{a _1}\rho {\Omega _2}}}}}} dz - \int_1^\infty  {\frac{1}{{{z^{n + 2}}}}{e^{ - \frac{{{a _2}z}}{{{a _1}\rho {\Omega _2}}}}}} dz \nonumber\\
  =&\frac{{{{\left( { - 1} \right)}^{2n + 1}}\phi _1^{^{n + 1}}}}{{(n + 1)!}}\left( {{\mathop{\rm Ei}\nolimits} \left( \psi  \right)} \right.\left. { - {\mathop{\rm Ei}\nolimits} \left( {{\phi _1}} \right)} \right) \nonumber \\
  &+ \sum\limits_{k = 0}^n  {\frac{{{{\left( {1 + {a _1}\rho \tau_{1} } \right)}^{n + 1}}{e^\psi }{\psi ^k} - {e^\phi }\phi _1^k}}{{\left( {n + 1} \right)n \cdots \left( {n + 1 - k} \right)}}}  ,
\end{align}
where \eqref{Theta2} can be obtained by using \cite[Eq. (3.351.4)]{gradshteyn}.

Substituting \eqref{Theta1} into \eqref{J3_1}, ${J_{11}}$ is written as
\begin{align}\label{J4_2}
 {J_{11}} =& \left\{ {1 - {e^{ - \frac{\tau_{1} }{{{\Omega _0}}}}} - \sum\limits_{n = 0}^\infty  {\frac{{{{\left( { - 1} \right)}^n}{e^\varphi }}}{{n!{\phi _2}^{n + 1}}}} } \right.\left[ {\frac{{{{\left( { - 1} \right)}^{2n + 1}}{\phi _1}^{n + 1}}}{{(n + 1)!}}\left( {{\mathop{\rm Ei}\nolimits} \left( \psi  \right)} \right.} \right. \nonumber \\
 &\left. {\left. {\left. { - {\mathop{\rm Ei}\nolimits} \left( {{\phi _1}} \right)} \right) + \sum\limits_{k = 0}^n  {\frac{{{{\left( {1 + {a _1}\rho \tau_{1} } \right)}^{n + 1}}{e^\psi }{\psi ^k} - {e^{{\phi _1}}}{\phi _1}^k}}{{\left( {n + 1} \right)n \cdots \left( {n + 1 - k} \right)}}} } \right]} \right\} .
\end{align}

After some algebraic manipulations, ${J_{12}}$ is calculated as
\begin{align}\label{J5}
 %{J_{12}} =& {{\rm{P}}_{\rm{r}}}\left( {{{\left| {{h_1}} \right|}^2} > \tau \left( {{{\left| {{h_{LI}}} \right|}^2}\rho  + 1} \right)} \right) \nonumber\\
 {J_{12}} =& \int_0^\infty  {\frac{1}{{{\Omega _{{\rm{LI}}}}}}{e^{ - \frac{y}{{{\Omega _{{\rm{LI}}}}}}}}\int_{\tau_{1} \left( {y\rho  + 1} \right)}^\infty  {\frac{1}{{{\Omega _1}}}{e^{ - \frac{x}{{{\Omega _1}}}}}} } dxdy \nonumber\\
  =& \chi {e^{ - \frac{\tau_{1} }{{{\Omega _1}}}}},
\end{align}
where $\chi  = \frac{{{\Omega _1}}}{{{\Omega _1} + \tau_{1} \rho {\Omega _{{\rm{LI}}}}}}$.

Similarly, ${J_{13}}$ is given by
\begin{align} \label{J3}
{J_{13}} =& \left( {1 - {e^{ - \frac{\tau_{1} }{{{\Omega _0}}}}}} \right)\left( {1 - \chi {e^{ - \frac{\tau_{1} }{{{\Omega _1}}}}}} \right).
\end{align}

Combining \eqref{J4_2}, \eqref{J5} and \eqref{J3}, we can obtain \eqref{OP derived for D2}.

The proof is completed.

\section*{Appendix~B: Proof of Theorem \ref{theorem_ergodic_rate_D1:4}} \label{Appendix:B}
%\section*{Appendix~A: Proof of Theorem \ref{theorem_outage:2}} \label{Appendix:A}
\renewcommand{\theequation}{B.\arabic{equation}}
\setcounter{equation}{0}
To obtain \eqref{ergodic rate of D1 without directlink for FD}, the ergodic rate of $D_{1}$ for FD NOMA is expressed as
\begin{align}\label{proof ergodic rate for D1}
R_{{D_1}}^{FD} =& \mathbb{E} \left[ {\log \left( {1 + \underbrace {\frac{{{{\left| {{h_1}} \right|}^2}{a _1}\rho }}{{\varpi {{\left| {{h_{LI}}} \right|}^2}\rho  + 1}}}_X} \right)} \right] \nonumber  \\
 =& \frac{1}{{\ln 2}}\int_0^\infty  {\frac{{1 - {F_X}\left( x \right)}}{{1 + x}}} dx
\end{align}
where $\varpi =1$.

The CDF of $X$ is calculated as follows
\begin{align}\label{the CDF of D1 for the ergodic rate}
{F_X}\left( x \right) =& {{\mathop{\rm P}\nolimits} _{\mathop{\rm r}\nolimits} }\left( {{{\left| {{h_1}} \right|}^2} < \frac{{x\left( {{{\left| {{h_{LI}}} \right|}^2}\rho  + 1} \right)}}{{{a_1}\rho }}} \right) \nonumber \\
  =& \int_0^\infty  {\frac{1}{{{\Omega _{LI}}}}{e^{ - \frac{z}{{{\Omega _{LI}}}}}}} \int_0^{\frac{{x\left( {z\rho  + 1} \right)}}{{{a_1}\rho }}} {\frac{1}{{{\Omega _1}}}{e^{ - \frac{y}{{{\Omega _1}}}}}} dydz \nonumber \\
  =&1 - \frac{{{a_1}{\Omega _1}}}{{x{\Omega _{LI}} + {a_1}{\Omega _1}}}{e^{ - \frac{x}{{{a_1}\rho {\Omega _1}}}}}.
\end{align}

Substituting \eqref{the CDF of D1 for the ergodic rate} into \eqref{proof ergodic rate for D1}, the ergodic rate of $D_{1}$ is written as
\begin{align}\label{the further expression of ergodic rate for D1}
 R_{{D_1}}^{FD} =& \frac{1}{{\ln 2}}\int_0^\infty  {\frac{1}{{1 + x}}} \frac{{{a_1}{\Omega _1}}}{{{a_1}{\Omega _1} + x{\Omega _{LI}}}}{e^{ - \frac{x}{{{a_1}\rho {\Omega _1}}}}}dx \nonumber\\
  =& \frac{1}{{\ln 2}}\underbrace {\int_0^\infty  {\frac{{ - {a_1}{\Omega _1}{e^{ - \frac{x}{{{a_1}\rho {\Omega _1}}}}}}}{{\left( {1 + x} \right)\left( {{\Omega _{LI}} - {a_1}{\Omega _1}} \right)}}} dx}_{{J_1}}\nonumber \\
  & + \frac{1}{{\ln 2}}\underbrace {\int_0^\infty  {\frac{{{a_1}{\Omega _1}{\Omega _{LI}}{e^{ - \frac{x}{{{a_1}\rho {\Omega _1}}}}}}}{{\left( {{a_1}{\Omega _1} + x{\Omega _{LI}}} \right)\left( {{\Omega _{LI}} - {a_1}{\Omega _1}} \right)}}} dx}_{{J_2}}.
\end{align}

Based on \cite[Eq. (3.352.4)]{gradshteyn} and applying some polynomial expansion manipulations, $J_{1}$ and $J_{2}$ are given by
\begin{align}\label{the further expression of ergodic rate for D1 J1}
{J_1} = \frac{{{a_1}{\Omega _1}{e^{\frac{1}{{{a_1}\rho {\Omega _1}}}}}}}{{{\Omega _{LI}} - {a_1}{\Omega _1}}}{{\mathop{\rm E}\nolimits} _{\mathop{\rm i}\nolimits} }\left( {\frac{{ - 1}}{{{a_1}\rho {\Omega _1}}}} \right),
\end{align}
and
\begin{align}\label{the further expression of ergodic rate for D1 J2}
{J_2} = \frac{{{a_1}{\Omega _1}{e^{\frac{1}{{\rho {\Omega _{LI}}}}}}}}{{{\Omega _{LI}} - {a _1}{\Omega _1}}}{{\mathop{\rm E}\nolimits} _{\mathop{\rm i}\nolimits} }\left( {\frac{{ - 1}}{{\rho {\Omega _{LI}}}}} \right).
\end{align}

Substituting \eqref{the further expression of ergodic rate for D1 J1} and \eqref{the further expression of ergodic rate for D1 J2} into \eqref{the further expression of ergodic rate for D1}, we can obtain \eqref{ergodic rate of D1 without directlink for FD}.

The proof is completed.

\section*{Appendix~C: Proof of Theorem \ref{theorem_ergodic_rate_D2:4}} \label{Appendix:C}
%\section*{Appendix~A: Proof of Theorem \ref{theorem_outage:2}} \label{Appendix:A}
\renewcommand{\theequation}{C.\arabic{equation}}
\setcounter{equation}{0}

The proof starts by providing the ergodic rate of $D_2$ as follows:
\begin{align}\label{proof ergodic rate for D2}
R_{{D_2},nodir}^{FD} =\mathbb{E}\left[ {\log \left( {1 + \underbrace {\min \left( {{\gamma _{{D_2} \to {D_1}}},{\gamma _{2,{D_2}}}} \right)}_{{J_1}}} \right)} \right],
\end{align}
where $\varpi  = 1$. We focus on the high SNR approximation of $J_{1}$, which is given by
\begin{align}\label{approximation J1}
{J_1} \approx \underbrace {\min \left( {\frac{{{{\left| {{h_1}} \right|}^2}{a _2}}}{{{{\left| {{h_1}} \right|}^2}{a _1} + {{\left| {{h_{LI}}} \right|}^2}}},{{\left| {{h_2}} \right|}^2}\rho} \right)}_Y.
\end{align}
The CDF of $Y$ is expressed as
\begin{align}\label{CDF of Y without direct link}
 {{\mathop{\rm F}\nolimits} _Y}\left( y \right) =& 1 - \underbrace {{{\rm{P}}_{\rm{r}}}\left( {{{\left| {{h_2}} \right|}^2}\rho > y} \right)}_{{\Theta _1}}  \underbrace {{{\rm{P}}_{\rm{r}}}\left( {\frac{{{{\left| {{h_1}} \right|}^2}{a _2}}}{{{{\left| {{h_1}} \right|}^2}{a _1} + {{\left| {{h_{LI}}} \right|}^2}}} > y} \right)}_{{\Theta _2}}.
\end{align}
%where ${\Theta _1} = {\mathop{\rm U}\nolimits} \left( {\frac{{{\alpha _2}}}{{{\alpha _1}}} - y} \right)\left( {1 - {e^{ - \frac{y}{{\rho {\Omega _2}}}}}} \right)$ can be easily obtained.

${\Theta _1}$ and ${\Theta _2}$ are given by
\begin{align}\label{Theta1 for ergodic rate without direct}
{\Theta _1} =1- {e^{ - \frac{y}{{\rho {\Omega _2}}}}}{\mathop{\rm U}\nolimits} \left( y \right),
\end{align}
and
%${{\Theta _2}}$ is calculated as follows:
\begin{align}\label{Theta2 for ergodic rate without direct}
{\Theta _2} = \frac{{\left( {{a_2} - y{a_1}} \right){\Omega _1}}}{{\left( {{a_2} - y{a_1}} \right){\Omega _1} + y{\Omega _{LI}}}}{\mathop{\rm U}\nolimits} \left( {\frac{{{a_2}}}{{{a_1}}} - y} \right),
\end{align}
respectively, where ${\mathop{\rm U}\nolimits} \left( x \right)$ is unit step function as\\ ${\mathop{\rm U}\nolimits} \left( x \right) = \left\{ {\begin{array}{*{20}{c}}
   {1\begin{array}{*{20}{c}}
   {,x > 0}  \\
\end{array}}  \\
   {0\begin{array}{*{20}{c}}
   {,x < 0}  \\
\end{array}}  \\
\end{array}} \right.$.

Substituting \eqref{Theta1 for ergodic rate without direct} and \eqref{Theta2 for ergodic rate without direct} into \eqref{CDF of Y without direct link}, the CDF of $Y$ is given by
\begin{align} \label{appro CDF of Y without direct link}
{{\mathop{\rm F}\nolimits} _Y}\left( y \right) = 1 - \frac{{{e^{ - \frac{y}{{\rho {\Omega _2}}}}}\left( {{a_2} - y{a_1}} \right){\Omega _1}}}{{\left( {{a_2} - y{a_1}} \right){\Omega _1} + y{\Omega _{LI}}}}{\mathop{\rm U}\nolimits} \left( y \right){\mathop{\rm U}\nolimits} \left( {\frac{{{a_2}}}{{{a_1}}} - y} \right).
\end{align}

Base on \eqref{appro CDF of Y without direct link}, a high SNR approximation of the ergodic rate for $D_{2}$ is written as
\begin{align}\label{temp appro ergodic rate for D2 without direct link}
 R_{{D_2},nodir}^{FD,\infty }%&= \frac{1}{{\ln 2}}\int_0^\infty  {\frac{{1 - {{\mathop{\rm F}\nolimits} _Y}\left( y \right)}}{{1 + y}}} dy \nonumber \\
  =& \frac{1}{{\ln 2}}\int_0^{\frac{{{a _2}}}{{{a _1}}}} {\frac{1}{{1 + y}}\frac{{{e^{ - \frac{y}{{\rho {\Omega _2}}}}}\left( {{a _2} - y{a _1}} \right){\Omega _1}}}{{\left( {{a _2} - y{a _1}} \right){\Omega _1} + y{\Omega _{LI}}}}} dy\nonumber  \\
  =& \frac{1}{{\ln 2}}\left( {\underbrace {\int_0^{\frac{{{a _2}}}{{{a _1}}}} {\frac{1}{{1 + y}}\frac{{{a _2}{\Omega _1}{e^{ - \frac{y}{{\rho {\Omega _2}}}}}}}{{y\xi  + {a _2}{\Omega _1}}}} dy}_{{J_2}}} \right. \nonumber \\
 &\left. { - \underbrace {\int_0^{\frac{{{a _2}}}{{{a _1}}}} {\frac{1}{{1 + y}}\frac{{y{a _1}{\Omega _1}{e^{ - \frac{y}{{\rho {\Omega _2}}}}}}}{{y\xi  + {a _2}{\Omega _1}}}} dy}_{{J_3}}} \right),
\end{align}
where $\xi  = \left( {{\Omega _{LI}} - {a _1}{\Omega _1}} \right)$.

Applying \cite[Eq. (3.352.1)]{gradshteyn} and some polynomial expansion manipulations, ${{J_2}}$ and ${{J_3}}$ can be calculated as
\begin{align}\label{appro J2 for ergodic without direct link}
 {J_2} &= \frac{{{a _2}{\Omega _1}}}{{{a _2}{\Omega _1} - \xi }}\left( {\int_0^{\frac{{{a _2}}}{{{a _1}}}} {\frac{{{e^{ - \frac{y}{{\rho {\Omega _2}}}}}}}{{1 + y}}dy - \int_0^{\frac{{{a _2}}}{{{a _1}}}} {\frac{{\xi {e^{ - \frac{y}{{\rho {\Omega _2}}}}}}}{{\xi y + {a _2}{\Omega _1}}}} } dy} \right) \nonumber \\
  &= \frac{{{a _2}{\Omega _1}}}{{{a _2}{\Omega _1} - \xi }}\left\{ {{e^{\frac{1}{{\rho {\Omega _2}}}}}\left[ {{\mathop{\rm Ei}\nolimits} \left( {\frac{{ - 1}}{{\rho {a _1}{\Omega _2}}}} \right) - {\mathop{\rm Ei}\nolimits} \left( {\frac{{ - 1}}{{\rho {\Omega _2}}}} \right)} \right]} \right.  \nonumber \\
  &- {e^{\frac{{{a _2}{\Omega _1}}}{{\rho {\Omega _2}\xi }}}}\left[ {{\mathop{\rm Ei}\nolimits} \left( { - \frac{{{a _2}\xi  + {a _1}{a _2}{\Omega _1}}}{{\rho {a _1}\xi {\Omega _2}}}} \right)} \right.\left. {\left. { - {\mathop{\rm Ei}\nolimits} \left( { - \frac{{{a _2}{\Omega _1}}}{{\rho {\Omega _2}\xi }}} \right)} \right]} \right\} .
\end{align}
\begin{align}\label{appro J3 for ergodic without direct link}
 {J_3} &= \frac{{{a _1}{\Omega _1}}}{{\xi  - {a _2}{\Omega _1}}}\left( {\int_0^{\frac{{{a _2}}}{{{a _1}}}} {\frac{{{e^{ - \frac{y}{{\rho {\Omega _2}}}}}}}{{1 + y}}dy - \int_0^{\frac{{{a _2}}}{{{a _1}}}} {\frac{{{a _2}{\Omega _1}{e^{ - \frac{y}{{\rho {\Omega _2}}}}}}}{{\xi y + {a _2}{\Omega _1}}}} } dy} \right) \nonumber\\
  &= \frac{{{a _1}{\Omega _1}}}{{\xi  - {a _2}{\Omega _1}}}\left\{ {{e^{\frac{1}{{\rho {\Omega _2}}}}}\left[ {{\mathop{\rm Ei}\nolimits} \left( {\frac{{ - 1}}{{\rho {a _1}{\Omega _2}}}} \right) - {\mathop{\rm Ei}\nolimits} \left( {\frac{{ - 1}}{{\rho {\Omega _2}}}} \right)} \right]} \right.  \nonumber\\
 &\left. { - \frac{{{a _2}{\Omega _1}{e^{\frac{{{a _2}{\Omega _1}}}{{\rho {\Omega _2}\xi }}}}}}{\xi }\left[ {{\mathop{\rm Ei}\nolimits} \left( { - \frac{{{a _2}\xi  + {a _1}{a _2}{\Omega _1}}}{{\rho {a _1}\xi {\Omega _2}}}} \right) - {\mathop{\rm Ei}\nolimits} \left( { - \frac{{{a _2}{\Omega _1}}}{{\rho {\Omega _2}\xi }}} \right)} \right]} \right\}.
\end{align}

Substituting \eqref{appro J2 for ergodic without direct link} and \eqref{appro J3 for ergodic without direct link} into \eqref{temp appro ergodic rate for D2 without direct link}, we can obtain \eqref{appro ergodic rate of D2 without directlink for FD}.

The proof is completed.

\section*{Appendix~D: Proof of Corollary \ref{high_SNR_sum_ergodic_rate_nodirectlink:4}} \label{Appendix:D}
\renewcommand{\theequation}{D.\arabic{equation}}
\setcounter{equation}{0}
We can rewrite \eqref{ergodic rate of D2 without directlink for HD} as follows:
\begin{align}
R_{{D_2},nodir}^{HD} = \frac{1}{2}\mathbb{E}\left[ {\log \left( {1 + \underbrace {\min \left( {{\gamma _{{D_2} \to {D_1}}},{\gamma _{2,{D_2}}}} \right)}_{{J_1}}} \right)} \right],
\end{align}
where $\varpi = 0$.

At the high SNR region, $J_{1}$ is approximated as
\begin{align}
{J_1} \approx \underbrace {\min \left( {\frac{{{a _2}}}{{{a _1}}},{{\left| {{h_2}} \right|}^2}\rho } \right)}_Y.
\end{align}
Therefore, ${{\mathop{\rm F}\nolimits} _Y}\left( y \right) = \left( {1 - {e^{ - \frac{y}{{\rho {\Omega _2}}}}}} \right){\mathop{\rm U}\nolimits} \left( {\frac{{{a_2}}}{{{a_1}}} - y} \right)$ can be easily obtained. As such, the approximated ergodic rate of $D_{2}$ for HD NOMA at high SNR is given in \eqref{appro ergodic rate of D2 without directlink for HD}.
%\begin{align}
 %R_{{D_2},ave}^{HD} &\approx \frac{1}{{2\ln 2}}\int_0^\infty  {\frac{{1 - {{\mathop{\rm F}\nolimits} _Y}\left( y \right)}}{{1 + y}}} dy \nonumber \\
 %  R_{{D_2},dir}^{HD,\infty } = \frac{{{e^{\frac{1}{{\rho {\Omega _2}}}}}}}{{2\ln 2}}\left[ {{\mathop{\rm Ei}\nolimits} \left( {\frac{{ - 1}}{{\rho {a _1}{\Omega _2}}}} \right) - {\mathop{\rm Ei}\nolimits} \left( {\frac{{ - 1}}{{\rho {\Omega _2}}}} \right)} \right].
%\end{align}

%\begin{center}
%APPENDIX D: PROOF OF THEOREM 5
%\end{center}
\section*{Appendix~E: Proof of Theorem \ref{theorem:6 direct link rate for D2}} \label{Appendix:E}
\renewcommand{\theequation}{E.\arabic{equation}}
\setcounter{equation}{0}
The proof starts by providing the ergodic rate of $D_2$ as follows:
\begin{align}\label{proof ergodic rate for D2}
R_{{D_2},dir}^{FD} = \mathbb{E}\left[ {\log \left( {1 + \underbrace {\min \left( {{\gamma _{{D_2} \to {D_1}}},\gamma _{{D_2}}^{MRC}} \right)}_{{J_1}}} \right)} \right],
\end{align}
where $\varpi  = 1$.
We focus on the high SNR approximation of $J_{1}$, which is given by
\begin{align}\label{approximation J1}
{J_1} \approx \underbrace {\min \left( {\frac{{{{\left| {{h_1}} \right|}^2}{a _2}}}{{{{\left| {{h_1}} \right|}^2}{a _1} + {{\left| {{h_{LI}}} \right|}^2}}},{{\left| {{h_2}} \right|}^2}\rho  + \frac{{{a _2}}}{{{a _1}}}} \right)}_Y.
\end{align}
The cumulative distribution function (CDF) of $Y$ is expressed as
\begin{align}\label{CDF of Y}
 {{\mathop{\rm F}\nolimits} _Y}\left( y \right) =& 1 - \underbrace {{{\rm{P}}_{\rm{r}}}\left( {{{\left| {{h_2}} \right|}^2}\rho  + \frac{{{a _2}}}{{{a _1}}} \ge y} \right)}_{{\Theta _1}} \nonumber \\
  &\times \underbrace {{{\rm{P}}_{\rm{r}}}\left( {\frac{{{{\left| {{h_1}} \right|}^2}{a _2}}}{{{{\left| {{h_1}} \right|}^2}{a _1} + {{\left| {{h_{LI}}} \right|}^2}}} \ge y} \right)}_{{\Theta _2}}.
\end{align}

${{\Theta _1}}$ and ${{\Theta _2}}$ are given by
\begin{align}\label{Theta1 for ergodic rate}
%{\Theta _1} = \left\{ {\begin{array}{*{20}{c}}
 %  {\begin{array}{*{20}{c}}
  % 1 & {}  \\
%\end{array}\begin{array}{*{20}{c}}
 %  {} & {}  \\
%\end{array}\begin{array}{*{20}{c}}
 %  {}  \\
%\end{array},y < \frac{{{\alpha _2}}}{{{\alpha _1}}}}  \\
 %  {{e^{ - \frac{1}{{\rho {\Omega _2}}}\left( {y - \frac{{{\alpha _2}}}{{{\alpha _1}}}} \right)}},y > \frac{{{\alpha _2}}}{{{\alpha _1}}}}  \\
%\end{array}} \right.
{\Theta _1} = 1 - {\mathop{\rm U}\nolimits} \left( {y - \frac{{{a _2}}}{{{a _1}}}} \right){e^{ - \frac{1}{{\rho {\Omega _2}}}\left( {y - \frac{{{a _2}}}{{{a _1}}}} \right)}},
\end{align}
and
\begin{align}\label{Theta2 for ergodic rate}
 {{\Theta _2}}=& {\mathop{\rm U}\nolimits} \left( {\frac{{{a _2}}}{{{a _1}}} - y} \right)\frac{{\left( {{a _2} - y{a _1}} \right){\Omega _1}}}{{\left( {{a _2} - y{a _1}} \right){\Omega _1} + y{\Omega _{LI}}}},
\end{align}
respectively.

Substituting \eqref{Theta1 for ergodic rate} and \eqref{Theta2 for ergodic rate} into \eqref{CDF of Y}, the CDF of $Y$ is given by
\begin{align} \label{appro CDF of Y}
%{F_Y}\left( y \right) = \left\{ {\begin{array}{*{20}{c}}
 %  {1 - \frac{{\left( {{\alpha _2} - y{\alpha _1}} \right){\Omega _1}}}{{\left( {{\alpha _2} - y{\alpha _1}} \right){\Omega _1} + y{\Omega _{LI}}}},y < \frac{{{\alpha _2}}}{{{\alpha _1}}}}  \\
 %  {1\begin{array}{*{20}{c}}
 %  {}  \\
%\end{array}\begin{array}{*{20}{c}}
 %  {} & {} & {}  \\
%\end{array}\begin{array}{*{20}{c}}
 %  {} & {} & {} & {}  \\
%\end{array},y > \frac{{{\alpha _2}}}{{{\alpha _1}}}}  \\
%\end{array}} \right.
{F_Y}\left( y \right) = 1 - {\mathop{\rm U}\nolimits} \left( {\frac{{{a _2}}}{{{a _1}}} - y} \right)\frac{{\left( {{a _2} - y{a _1}} \right){\Omega _1}}}{{\left( {{a _2} - y{a _1}} \right){\Omega _1} + y{\Omega _{LI}}}}.
\end{align}

Base on \eqref{appro CDF of Y}, a high SNR approximation of the ergodic rate for $D_{2}$ is written as
\begin{align}\label{temp appro ergodic rate for D2}
 R_{{D_2},dir}^{FD,\infty } %&= \frac{1}{{\ln 2}}\int_0^\infty  {\frac{{1 - {F_Y}\left( y \right)}}{{1 + y}}} dy \nonumber \\
  =& \frac{1}{{\ln 2}}\int_0^{\frac{{{a _2}}}{{{a _1}}}} {\frac{1}{{1 + y}}\frac{{\left( {{a _2} - y{a _1}} \right){\Omega _1}}}{{\left( {{a _2} - y{a _1}} \right){\Omega _1} + y{\Omega _{LI}}}}} dy \nonumber\\
  =& \frac{1}{{\ln 2}}\left( {\underbrace {\int_0^{\frac{{{a _2}}}{{{a _1}}}} {\frac{{{a _2}{\Omega _1}}}{{\left( {1 + y} \right)\left( {y\xi  + {a _2}{\Omega _1}} \right)}}} dy}_{{J_2}}} \right. \nonumber\\
& \left. { - \underbrace {\int_0^{\frac{{{a _2}}}{{{a _1}}}} {\frac{{y{a _1}{\Omega _1}}}{{\left( {1 + y} \right)\left( {y\xi  + {a _2}{\Omega _1}} \right)}}} }_{{J_3}}{\rm{dy}}} \right),
\end{align}
where $\xi  = \left( {{\Omega _{LI}} - {a _1}{\Omega _1}} \right)$.

After some algebraic manipulations, ${{J_2}}$ and ${{J_3}}$ are obtained as follows:
\begin{align}\label{appro J2 for ergodic}
 {J_2} %=& \frac{{{\alpha _2}{\Omega _1}}}{{{\alpha _2}{\Omega _1} - \xi }}\left( {\int_0^{\frac{{{\alpha _2}}}{{{\alpha _1}}}} {\frac{1}{{1 + x}}dx - \int_0^{\frac{{{\alpha _2}}}{{{\alpha _1}}}} {\frac{\xi }{{x\xi  + {\alpha _2}{\Omega _1}}}} } dx} \right) \nonumber \\
  =& \frac{{{a _2}{\Omega _1}}}{{{a _2}{\Omega _1} - \xi }}\left[ {\ln \left( {1 + \frac{{{a _2}}}{{{a _1}}}} \right) - \ln \left( {1 + \frac{\xi }{{{a _1}{\Omega _1}}}} \right)} \right].
\end{align}
\begin{align}\label{appro J3 for ergodic}
 {J_3} %=& \frac{{{\alpha _1}{\Omega _1}}}{{{\alpha _2}{\Omega _1} - \xi }}\left( {\int_0^{\frac{{{\alpha _2}}}{{{\alpha _1}}}} {\frac{{{\alpha _2}{\Omega _1}}}{{y\xi  + {\alpha _2}{\Omega _1}}}dy - } \int_0^{\frac{{{\alpha _2}}}{{{\alpha _1}}}} {\frac{{\rm{1}}}{{1 + y}}dy} } \right) \nonumber \\
  =& \frac{{{a _1}{\Omega _1}}}{{{a _2}{\Omega _1} - \xi }}\left[ {\frac{{{a _2}{\Omega _1}}}{\xi }\ln \left( {1 + \frac{\xi }{{{a _1}{\Omega _1}}}} \right) - \ln \left( {1 + \frac{{{a _2}}}{{{a _1}}}} \right)} \right].
\end{align}

Substituting \eqref{appro J2 for ergodic} and \eqref{appro J3 for ergodic} into \eqref{temp appro ergodic rate for D2}, we can obtain \eqref{appro ergodic rate of D2 with directlink for FD}.

The proof is completed.
\bibliographystyle{IEEEtran}
\bibliography{mybib}

\begin{thebibliography}{10}
\providecommand{\url}[1]{#1}
\csname url@samestyle\endcsname
\providecommand{\newblock}{\relax}
\providecommand{\bibinfo}[2]{#2}
\providecommand{\BIBentrySTDinterwordspacing}{\spaceskip=0pt\relax}
\providecommand{\BIBentryALTinterwordstretchfactor}{4}
\providecommand{\BIBentryALTinterwordspacing}{\spaceskip=\fontdimen2\font plus
\BIBentryALTinterwordstretchfactor\fontdimen3\font minus
  \fontdimen4\font\relax}
\providecommand{\BIBforeignlanguage}[2]{{%
\expandafter\ifx\csname l@#1\endcsname\relax
\typeout{** WARNING: IEEEtran.bst: No hyphenation pattern has been}%
\typeout{** loaded for the language `#1'. Using the pattern for}%
\typeout{** the default language instead.}%
\else
\language=\csname l@#1\endcsname
\fi
#2}}
\providecommand{\BIBdecl}{\relax}
\BIBdecl

\bibitem{Yue2016Non}
X.~Yue, Y.~Liu, S.~Kang, A.~Nallanathan, and Z.~Ding, ``Outage performance of
  full/half-duplex user relaying in {NOMA} systems,'' in \emph{IEEE Proc. of
  International Commun. Conf. (ICC)}, Paris, FRA, May. 2017, pp. 1--6.

\bibitem{Cai2017Modulation}
\BIBentryALTinterwordspacing
Y.~Cai, Z.~Qin, F.~Cui, G.~Y. Li, and J.~A. McCann, ``Modulation and multiple
  access for 5{G} networks,'' CoRR, vol. abs/1702.07673, 2017. [Online].
  Available: \url{https://arxiv.org/abs/1702.07673}
\BIBentrySTDinterwordspacing

\bibitem{METIS}
``¡°{P}roposed solutions for new radio access,¡± {M}obile and wireless
  communications {E}nablers for the {T}wenty-twenty {I}nformation {S}ociety
  {(METIS)}, {D}eliverable {D}.2.4, {F}eb. 2015.''

\bibitem{Saito2013Non}
Y.~Saito, Y.~Kishiyama, A.~Benjebbour, T.~Nakamura, A.~Li, and K.~Higuchi,
  ``Non-orthogonal multiple access ({NOMA}) for cellular future radio access,''
  in \emph{Proc. IEEE Vehicular Technology Conference (VTC Spring)}, Dresden,
  GRE, Jun. 2013, pp. 1--5.

\bibitem{Nikopour2013Sparse}
H.~Nikopour and H.~Baligh, ``Sparse code multiple access,'' in \emph{Proc. IEEE
  Annual International Symposium on Personal, Indoor, and Mobile Radio
  Communications (PIMRC)}, London, UK, Sep. 2013, pp. 332--336.

\bibitem{ChenPattern7526461}
S.~Chen, B.~Ren, Q.~Gao, S.~Kang, S.~Sun, and K.~Niu, ``Pattern division
  multiple access {PDMA} - {A} novel nonorthogonal multiple access for
  fifth-generation radio networks,'' \emph{{IEEE} Trans. Veh. Technol.},
  vol.~66, no.~4, pp. 3185--3196, Apr. 2017.

\bibitem{Yuan2016Multi}
Z.~Yuan, G.~Yu, W.~Li, Y.~Yuan, X.~Wang, and J.~Xu, ``Multi-user shared access
  for internet of things,'' in \emph{Proc. IEEE Vehicular Technology Conference
  (VTC Spring)}, Nanjing, CHN, May. 2016, pp. 1--5.

\bibitem{MUST1}
``3rd {G}eneration {P}artnership {P}rojet (3{GPP}), "{S}tudy on downlink
  multiuser superposition transmation for {LTE}," {M}ar. 2015.''

\bibitem{Ding2015Application}
Z.~Ding, Y.~Liu, J.~Choi, Q.~Sun, M.~Elkashlan, I.~Chih-Lin, and H.~V. Poor,
  ``Application of non-orthogonal multiple access in {LTE} and {5G} networks,''
  \emph{{IEEE} Commun. Mag.}, vol.~55, no.~2, pp. 185--191, Feb. 2017.

\bibitem{Al6933459}
M.~Al-Imari, P.~Xiao, M.~A. Imran, and R.~Tafazolli, ``Uplink non-orthogonal
  multiple access for {5G} wireless networks,'' in \emph{Proc. IEEE
  International Symposium on Wireless Communications Systems (ISWCS)},
  Barcelona, ESP, Aug. 2014, pp. 781--785.

\bibitem{Zhang7390209}
N.~Zhang, J.~Wang, G.~Kang, and Y.~Liu, ``Uplink nonorthogonal multiple access
  in {5G} systems,'' \emph{{IEEE} Commun. Lett.}, vol.~20, no.~3, pp. 458--461,
  Mar. 2016.

\bibitem{Ding6868214}
Z.~Ding, Z.~Yang, P.~Fan, and H.~V. Poor, ``On the performance of
  non-orthogonal multiple access in 5{G} systems with randomly deployed
  users,'' \emph{{IEEE} Signal Process. Lett.}, vol.~21, no.~12, pp.
  1501--1505, Dec. 2014.

\bibitem{Pairing7273963}
Z.~Ding, P.~Fan, and H.~V. Poor, ``Impact of user pairing on 5{G}
  non-orthogonal multiple-access downlink transmissions,'' \emph{{IEEE} Trans.
  Veh. Technol.}, vol.~65, no.~8, pp. 6010--6023, Aug. 2016.

\bibitem{Liu2016TVT}
Y.~Liu, Z.~Ding, M.~Elkashlan, and J.~Yuan, ``Non-orthogonal multiple access in
  large-scale underlay cognitive radio networks,'' \emph{{IEEE} Trans. Veh.
  Technol.}, vol.~65, no.~12, pp. 10\,152--10\,157, Dec. 2016.

\bibitem{Tian7390804}
Y.~Tian, S.~Lu, A.~Nix, and M.~Beach, ``A novel opportunistic {NOMA} in
  downlink coordinated multi-point networks,'' in \emph{Proc. IEEE Vehicular
  Technology Conference (VTC Fall)}, Boston, USA, Sep. 2015, pp. 1--5.

\bibitem{Yang2016Outage}
Z.~Yang, Z.~Ding, P.~Fan, and Z.~Ma, ``Outage performance for dynamic power
  allocation in hybrid non-orthogonal multiple access systems,'' \emph{{IEEE}
  Commun. Lett.}, vol.~20, no.~8, pp. 1695--1698, Aug. 2016.

\bibitem{Botsinis2016Quantum}
P.~Botsinis, D.~Alanis, Z.~Babar, H.~Nguyen, D.~Chandra, S.~X. Ng, and
  L.~Hanzo, ``Quantum-aided multi-user transmission in non-orthogonal multiple
  access systems,'' \emph{IEEE Access}, to appear in 2016.

\bibitem{Choi7383267}
J.~Choi, ``On the power allocation for a practical multiuser superposition
  scheme in {NOMA} systems,'' \emph{{IEEE} Commun. Lett.}, vol.~20, no.~3, pp.
  438--441, Mar. 2016.

\bibitem{Physical7812773}
Y.~Liu, Z.~Qin, M.~Elkashlan, Y.~Gao, and L.~Hanzo, ``Enhancing the physical
  layer security of non-orthogonal multiple access in large-scale networks,''
  \emph{{IEEE} Trans. Wireless Commun.}, vol.~16, no.~3, pp. 1656--1672, Mar.
  2017.

\bibitem{laneman2004cooperative}
J.~N. Laneman, D.~N.~C. Tse, and G.~W. Wornell, ``Cooperative diversity in
  wireless networks: Efficient protocols and outage behavior,'' \emph{{IEEE}
  Trans. Inf. Theory}, vol.~50, no.~12, pp. 3062--3080, Dec. 2004.

\bibitem{Choi6708131}
J.~Choi, ``Non-orthogonal multiple access in downlink coordinated two-point
  systems,'' \emph{{IEEE} Commun. Lett.}, vol.~18, no.~2, pp. 313--316, Feb.
  2014.

\bibitem{Kim7230246}
J.~B. Kim and I.~H. Lee, ``Non-orthogonal multiple access in coordinated direct
  and relay transmission,'' \emph{{IEEE} Commun. Lett.}, vol.~19, no.~11, pp.
  2037--2040, Nov. 2015.

\bibitem{Kim2015Capacity}
------, ``Capacity analysis of cooperative relaying systems using
  non-orthogonal multiple access,'' \emph{{IEEE} Commun. Lett.}, vol.~19,
  no.~11, pp. 1949--1952, Nov. 2015.

\bibitem{Men7454773}
J.~Men, J.~Ge, and C.~Zhang, ``Performance analysis of nonorthogonal multiple
  access for relaying networks over nakagami-$m$ fading channels,''
  \emph{{IEEE} Trans. Veh. Technol.}, vol.~66, no.~2, pp. 1200--1208, Feb.
  2017.

\bibitem{Ding2014Cooperative}
Z.~Ding, M.~Peng, and H.~V. Poor, ``Cooperative non-orthogonal multiple access
  in 5{G} systems,'' \emph{{IEEE} Commun. Lett.}, vol.~19, no.~8, pp.
  1462--1465, Aug. 2015.

\bibitem{Liu7445146SWIPT}
Y.~Liu, Z.~Ding, M.~Elkashlan, and H.~V. Poor, ``Cooperative non-orthogonal
  multiple access with simultaneous wireless information and power transfer,''
  \emph{{IEEE} J. Sel. Areas Commun.}, vol.~34, no.~4, pp. 938--953, Apr. 2016.

\bibitem{Ju2009Improving}
H.~Ju, E.~Oh, and D.~Hong, ``Improving efficiency of resource usage in two-hop
  full duplex relay systems based on resource sharing and interference
  cancellation,'' \emph{{IEEE} Trans. Wireless Commun.}, vol.~8, no.~8, pp.
  3933--3938, Aug. 2009.

\bibitem{Riihonen2011Mitigation}
T.~Riihonen, S.~Werner, and R.~Wichman, ``Mitigation of loopback
  self-interference in full-duplex {MIMO} relays,'' \emph{{IEEE} Trans. Signal
  Process.}, vol.~59, no.~12, pp. 5983--5993, Dec. 2011.

\bibitem{Zhang2015Full}
Z.~Zhang, X.~Chai, K.~Long, A.~V. Vasilakos, and L.~Hanzo, ``Full duplex
  techniques for 5{G} networks: self-interference cancellation, protocol
  design, and relay selection,'' \emph{{IEEE} Commun. Mag.}, vol.~53, no.~5,
  pp. 128--137, May. 2015.

\bibitem{Wang2015Outage}
Q.~Wang, Y.~Dong, X.~Xu, and X.~Tao, ``Outage probability of full-duplex {AF}
  relaying with processing delay and residual self-interference,'' \emph{{IEEE}
  Commun. Lett.}, vol.~19, no.~5, pp. 783--786, May. 2015.

\bibitem{Exploiting7105858}
D.~P.~M. Osorio, E.~E.~B. Olivo, H.~Alves, J.~C. S.~S. Filho, and M.~Latva-aho,
  ``Exploiting the direct link in full-duplex amplify-and-forward relaying
  networks,'' \emph{{IEEE} Signal Process. Lett.}, vol.~22, no.~10, pp.
  1766--1770, Oct. 2015.

\bibitem{Kwon2010Optimal}
T.~Kwon, S.~Lim, S.~Choi, and D.~Hong, ``Optimal duplex mode for {DF} relay in
  terms of the outage probability,'' \emph{{IEEE} Trans. Veh. Technol.},
  vol.~59, no.~7, pp. 3628--3634, Sep. 2010.

\bibitem{Riihonen5961159Hybrid}
T.~Riihonen, S.~Werner, and R.~Wichman, ``Hybrid full-duplex/half-duplex
  relaying with transmit power adaptation,'' \emph{{IEEE} Trans. Wireless
  Commun.}, vol.~10, no.~9, pp. 3074--3085, Sep. 2011.

\bibitem{Ding2016FD}
Z.~Zhang, Z.~Ma, M.~Xiao, Z.~Ding, and P.~Fan, ``Full-duplex device-to-device
  aided cooperative non-orthogonal multiple access,'' \emph{{IEEE} Trans. Veh.
  Technol.}, vol.~66, no.~5, pp. 4467--4471, May. 2017.

\bibitem{Krikidis2012Full}
I.~Krikidis, H.~A. Suraweera, P.~J. Smith, and C.~Yuen, ``Full-duplex relay
  selection for amplify-and-forward cooperative networks,'' \emph{{IEEE} Trans.
  Wireless Commun.}, vol.~11, no.~12, pp. 4381--4393, Dec. 2012.

\bibitem{Cover1991Elements}
T.~M. Cover and J.~A. Thomas, \emph{Elements of information theory}, 6th~ed.,
  Wiley and Sons, New York, 1991.

\bibitem{Zhong6898012}
C.~Zhong, H.~A. Suraweera, G.~Zheng, I.~Krikidis, and Z.~Zhang, ``Wireless
  information and power transfer with full duplex relaying,'' \emph{{IEEE}
  Trans. Commun.}, vol.~62, no.~10, pp. 3447--3461, Oct. 2014.

\bibitem{gradshteyn}
I.~S. Gradshteyn and I.~M. Ryzhik, \emph{Table of Integrals, Series and
  Products}, 6th~ed.\hskip 1em plus 0.5em minus 0.4em\relax New York, NY, USA:
  Academic Press, 2000.

\bibitem{Bai7097002}
Z.~Bai, J.~Jia, C.~X. Wang, and D.~Yuan, ``Performance analysis of snr-based
  incremental hybrid decode-amplify-forward cooperative relaying protocol,''
  \emph{{IEEE} Trans. Commun.}, vol.~63, no.~6, pp. 2094--2106, Jun. 2015.

\end{thebibliography}

\end{document}